%% file: main.tex
\documentclass[11pt]{article}
\usepackage{fullpage}
\usepackage{authblk}
\usepackage{float}
\usepackage[hidelinks]{hyperref}
\usepackage{amsmath}
\usepackage{amsthm}
\usepackage[linesnumbered, noend, noline, ruled]{algorithm2e}
\SetKwComment{Comment}{$\triangleright$\ }{}

\usepackage{thm-restate}
\usepackage{enumitem}
\usepackage[dvipsnames]{xcolor}
\usepackage[numbers]{natbib}
\usepackage{appendix}

\input{notation.tex}

\newtheorem{theorem}{Theorem}
\newtheorem{definition}{Definition}
\newtheorem{lemma}{Lemma}

\title{EFx Budget-Feasible Allocations with High Nash Welfare}
\date{}
\author{Marius Garbea\thanks{mgarbea@drexel.edu}}
\author{Vasilis Gkatzelis\thanks{gkatz@drexel.edu}}
\author{Xizhi Tan\thanks{xizhi@drexel.edu}}
\affil{Drexel University, Computer Science}
\begin{document}

\maketitle

\input{abstract}

\section{Introduction}
\input{introduction}\label{sec:introduction}

\textbf{Additional related work.}
\input{relatedWork}\label{sec:related-work}

\section{Preliminaries}
\input{preliminaries}\label{sec:preliminaries}

\section{Instances Involving Two Agents}\label{sec:2agents-instance}
\input{EF1andPO}

\subsection{EFx Allocations for Two-Agent Instances}\label{sec:2agents-alg}
\input{2agents}

\section{EFx Allocations for 3-Agent Instances}\label{sec:3agents-alg}
\input{3agents}

\section{Conclusion and Open Problems}
\input{conclusion}\label{sec:conclusion}

\section*{Acknowledgements}
The authors were partially supported by NSF CAREER award CCF 2047907.

\newpage
\bibliographystyle{plainnat}
\bibliography{references}

\newpage
\appendix
\section*{APPENDIX}
\section*{Complete proofs of Lemmas~\ref{lemma:equalbudget}, \ref{lemma:return1-else}, \ref{lemma:return2-else}, \ref{lemma:return3-else}}
\input{appendix}

\end{document}

%% file: notation.tex
\DeclareMathOperator*{\argmax}{arg\!\max}
\DeclareMathOperator*{\argmin}{arg\!\min}

\newcommand{\set}[1]{\left\{#1\right\}}
\newcommand{\alloc}{\mathbf{X}}
\newcommand{\bundle}{\mathrm{X}}
\newcommand{\allocSet}{\mathcal{X}}
\newcommand{\allocALG}{\alloc^{\texttt{ALG}}}
\newcommand{\allocALGBundle}{\bundle^{\texttt{ALG}}}
\newcommand{\OPT}{\alloc^{\texttt{OPT}}}
\newcommand{\OPTBundle}{\bundle^{\texttt{OPT}}}
\newcommand{\NSW}{\texttt{NSW}}
\newcommand{\twoalgoDB}{\texttt{EFx-2A}}
\newcommand{\threealgoDB}{\texttt{EFx-3A}}
\newcommand{\setaside}{s}
\newcommand{\threshold}{\alpha}
\newcommand{\vmax}{v^{\max}}
\newcommand{\vargmax}{S^{\max}}
\newcommand{\EFxGraph}{G_{EFx}}
\newcommand{\leximin}{\texttt{leximin++}}

%% file: abstract.tex
\begin{abstract}
    We study the problem of allocating indivisible items to budget-constrained agents, aiming to provide fairness and efficiency guarantees. Specifically, our goal is to ensure that the resulting allocation is envy-free up to any item (EFx) while minimizing the amount of inefficiency that this needs to introduce. We first show that there exist two-agent problem instances for which no EFx allocation is Pareto-efficient. We, therefore, turn to approximation and use the (Pareto-efficient) maximum Nash welfare allocation as a benchmark. For two-agent instances, we provide a procedure that always returns an EFx allocation while achieving the best possible approximation of the optimal Nash social welfare that EFx allocations can achieve. For the more complicated case of three-agent instances, we provide a procedure that guarantees EFx, while achieving a constant approximation of the optimal Nash social welfare for any number of items.
\end{abstract}

%% file: introduction.tex
One of the central open challenges in the fair division literature is the (approximately) envy-free allocation of indivisible goods \cite{ALMW2022, ABRV2022}, i.e., goods that cannot be shared among multiple agents. Envy-freeness is a very natural and well-motivated goal: an allocation of goods among a group of agents is envy-free when no agent would prefer the set of goods allocated to some other agents over the ones allocated to her. However, the fundamental challenge that arises when the goods are indivisible is that envy-freeness may only be achievable by discarding all the goods, which provides no value to any of the agents. An illustrative example of this fact arises when two agents compete over a single indivisible good: if the good is allocated to any of the agents, the other agent is bound to envy her.

To overcome this obstacle, the currently most vibrant line of research in fair division focuses on relaxations of envy-freeness, aiming to provide approximations of this natural fairness property without sacrificing too much value. Two such relaxations that have dominated these efforts are \emph{envy-freeness up to some good} (EF1) (defined explicitly by Budish \cite{Budish2010} and implicitly in the earlier work of Lipton et al. \cite{LMMS2004}) and the more demanding \emph{envy-freeness up to any good} (EFx) (defined by Caragiannis et al. \cite{CKMPSW2019}). Both of these notions ensure that any envy from agent $i$ toward agent $j$ would disappear if $i$ were to remove just one of the goods allocated to agent $j$ (her favorite good in the case of EF1 and \emph{any} good in the case of EFx). Therefore, any envy that may exist in EF1 and EFx allocations is up to just one good, which nicely sidesteps the aforementioned single-good illustrative example. In fact, for the well-studied case where the agents' valuations are additive across the goods (i.e., each agent $i$ has a value $v_i(g)$ for each good $g$ and a value $v_i(S)=\sum_{g\in S} v_i(g)$ for any set $S$ of goods), a notable result by Caragiannis et al. \cite{CKMPSW2019} shows that returning the highly appealing allocation that maximizes the Nash social welfare ($\NSW$), i.e., the geometric mean of the agents' values, always satisfies EF1. Therefore, unlike envy-freeness, EF1 can be combined with Pareto efficiency.
EFx allocations, on the other hand, are much more elusive: if all the goods need to be allocated, then for instances involving more than three agents we do not even know if EFx allocations exist. For instances with up to three agents, Chaudhury et al. \cite{CGM2020} proved the existence of EFx allocations using a highly non-trivial procedure that does not guarantee Pareto efficiency. 

Rather than assuming that all the goods need to be allocated, a recent line of work has instead proposed procedures that may donate some of the goods, as long as the allocation of the remaining goods (i.e., the ones not donated) satisfies EFx as well as some approximate efficiency guarantees \cite{CGH2019,CKMS2021,BCFF2022}. Notably, using this approach for the case of additive valuations, Caragiannis et al. \cite{CGH2019} were able to produce EFx allocations whose Nash social welfare is at least half of the optimal Nash welfare. Therefore, these allocations simultaneously guarantee fairness, in the form of EFx, and efficiency, by approximating a highly desirable Pareto-efficient allocation (unlike other Pareto-efficient allocations, the one maximizing the Nash social welfare is known to provide a natural balance between fairness and efficiency). In this paper, we study the extent to which analogous results that combine fairness and efficiency can be achieved beyond the special case of additive valuations, and we study the more demanding setting where the agents face budget constraints. 

The fair allocation of indivisible goods among agents with additive valuations but hard budget constraints was first studied by Wu et al. \cite{WLG2021}. In this setting, each good $g$ is associated with a cost $c(g)$, and each agent $i$ has a budget $B_i$ which restricts her to allocations with a total cost within her budget (i.e., she can be allocated a set of goods $S$ only if $\sum_{g\in S} c(g) \leq B_i$). These ``budgets'' can correspond to actual monetary budgets, but they can also capture several natural space or time limitations beyond that (e.g., when each agent faces, possibly different, capacity constraints for storing the goods allocated to her). Wu et al. \cite{WLG2021} showed that maximizing the Nash welfare subject to these budget constraints does not guarantee EF1, in contrast to the setting without budgets~\cite{CKMPSW2019}. Restricting their attention to the Nash welfare maximizing allocation, they showed that it is approximately EF1 in the sense that agent $i$'s value for $j$'s bundle after removing a good from it can be no more than 4 times her value for her own bundle. Furthermore, they showed that this approximate EF1 bound is tight. Subsequent work on this setting remained focused on the same allocation, proving improved approximate EF1 guarantees for the special case where all agents have identical valuations \cite{GLW2021}, or approximate-EFx guarantees for the special case where the values are all binary (i.e., $v_i(g)\in \set{0,1}$ for all $g$) \cite{DGLLNXZ2022}. Very recent work by Barman et al. \cite{BKSS2022} has shown the existence of EF2 allocations (in which no agent envies another agent if he could remove 2 goods from everyone's bundle) even in the budgeted setting.

\vspace{8pt}
\textbf{Our Results.}
We study the fair allocation of indivisible goods among agents with budget constraints and rather than restricting our attention to the maximum Nash welfare allocation, which does not guarantee EF1 or EFx, we instead guarantee EFx exactly and follow the approach of Caragiannis et al. \cite{CGH2019}. That is, we do not assume that all goods need to be allocated (which, in fact, may be infeasible in a setting with budgets) and instead provide efficiency guarantees by proving that our allocations approximate the Nash welfare optimal outcome, which is Pareto-efficient. In the presence of budget constraints, even achieving EF1 without sacrificing too much efficiency becomes a non-trivial problem, which is in contrast to the trivial ways in which EF1 can be achieved in the unrestricted additive valuations case (e.g., using a round-robin procedure). In fact, our first result shows that even for instances involving just three goods and two agents with equal budgets, simultaneously guaranteeing EF1 and Pareto efficiency is infeasible. We then complement this negative result with two positive results for instances involving an arbitrary number of goods and two or three agents of arbitrary budgets.

In Section~\ref{sec:2agents-instance} we focus on instances involving two agents and provide a procedure that returns an EFx allocation with a $\sqrt{0.5}\approx 0.7$ approximation of the optimal Nash welfare. We then show this procedure is optimal in a strong sense: it achieves the best possible approximation that one can guarantee not just for EFx allocations, but even for the more permissive EF1 guarantee. This procedure starts from the budget-feasible allocation with the optimal Nash welfare and, if this is not EFx already, it partitions the bundle of the ``envied'' agent into two sub-bundles, so that an EFx allocation can be achieved by matching the two agents to two of the three bundles.

In Section~\ref{sec:3agents-alg} we study instances involving three agents. Even if all these agents have the same budget, this is significantly more complicated than unrestricted additive valuations because of the limited ways in which goods can be reallocated across agents' bundles without violating budget feasibility. When each agent's budget is different, this obstacle becomes even more pronounced, as one agent's feasible bundle may be infeasible for another. We provide a procedure that first considers the extreme solution of letting the agents arrive in increasing order of their budgets and choosing their favorite budget-feasible bundle among the remaining goods. We then observe that if in the resulting allocation an agent $i$ envies some agent $j$ who arrived earlier, this allows us to lower $i$'s budget to be equal to agent $j$'s, without sacrificing too much efficiency. Using this intuition, our procedure either reaches an EFx allocation with different budgets or reduces the problem to an instance with equal budgets. Our main result in this section shows that the EFx allocations returned by our procedure always guarantee a constant approximation of the optimal Nash welfare (i.e., an approximation that does not grow with the number of goods or the size of the budgets).

%% file: relatedWork.tex
Following the approach introduced by Caragiannis et al. \cite{CGH2019}, i.e., donating goods to limit the envy, recent work has also aimed to achieve EFx while minimizing the \emph{number} of donated goods \cite{CKMS2021, BCFF2022}. 
This is in contrast to the results of Caragiannis et al. \cite{CGH2019} and our results in this paper, which do not optimize for the number of goods donated, as long as the resulting allocation is approximately efficient. Note that minimizing the number of donated goods does not provide any efficiency guarantees in general (e.g., even donating a single good could lead to very low efficiency if that was a highly valued good). Some of this work also ensures that no agent envies the set of donated, which limits the amount of donated value. In fact, even an EFx allocation that donates none of the goods is not necessarily efficient: e.g., if each agent values a distinct subset of goods, an allocation that partitions each such subset ``equally'' among all the agents is envy-free, yet inefficient. For the special case where all the agents have just two values for the goods and additive values across goods, Amanatidis et al. \cite{ABFHV2021} showed that maximizing the Nash welfare does yield EFx allocations. In general, however, efficiency and fairness cannot be fully achieved simultaneously.

Chaudhury et al. \cite{CGM2021} consider the more general class of subadditive valuations and focus on allocations that are approximately, rather than exactly, EFx and that approximate the Nash social welfare optimal outcome. Feldman et al. \cite{FMP2023} then capture the tension between EFx and the Nash social welfare by providing tight bounds regarding the trade-off of the approximations achievable for these two notions, both for additive and subadditive valuations.

%% file: preliminaries.tex
Given a set of indivisible goods $M = \set{1,\ldots,m}$, we seek to allocate (a subset of) these goods to a group of agents $N = \set{1,\ldots,n}$, where each good $g \in M$ has a cost $c(g)\geq 0$ and each agent $i \in N$ has a budget $B_i\geq 0$. For each subset of goods $S \subseteq M$, we let $c(S) = \sum_{g \in S} c(g)$ represent the total cost of $S$ and say that $S$ is \emph{budget-feasible for agent $i$} if $c(S)\leq B_i$. Each agent $i$ has a value $v_i(g)\geq 0$ for each good $g$ and her value for being allocated a budget-feasible subset of goods $S\subseteq M$ is additive, i.e., $v_i(S) = \sum_{g \in S} v_i(g)$. Given a set of goods $G$ that may not be budget-feasible for agent $i$, we use 
$\vmax_i(G)= \max_{S \subseteq G:c(S) \leq B_i} v_i(S)$ to denote the maximum value that the agent can achieve through a subset $S$ of $G$ that is budget-feasible for her. We also use $\vargmax_i(G)$ to denote the subset of goods in $G$ that is budget-feasible for $i$ and achieves the maximum value, i.e. $\vmax_i(G) = v_i(\vargmax_i(G))$.

An allocation $\alloc=(\bundle_1, \ldots, \bundle_n)$ determines what subset of goods $\bundle_i\subseteq M$ each agent $i\in N$ gets. Since the goods are indivisible, these sets need to be disjoint, i.e., $\bundle_i \cap \bundle_j=\emptyset$ for all $i,j\in N$. We say that an \emph{allocation is budget-feasible} if $c(X_i)\leq B_i$ for all $i\in N$. Note that we do not require an allocation $\alloc$ to be complete (i.e., that every good is allocated to some agent); in fact, it may be infeasible to do that in a budget-feasible allocation. 

Given the agents' budget constraints, an allocation $\alloc$ is \emph{envy-free} if no agent $i$ would improve her value by replacing their bundle with a subset of another agent's bundle that is budget-feasible for $i$. Formally, for any two agents $i, j \in N$ 
we have $v_i(\bundle_i)\geq \vmax_i(\bundle_j)$,
so $i$ has no ``justifiable complaints.'' It is well-known that, even in the absence of any budget constraints, it may be impossible to allocate indivisible goods in an envy-free way, while ensuring that agents receive any positive value.\footnote{For example, consider a simple instance with $n=2$ agents and $m=2$ goods, such that both agents $i\in \set{1,2}$ value the first good $v_i(1)=100$ and the second good $v_i(2)=50$. If any agent is allocated the first good, the other agent will envy them, so the first good needs to remain unallocated. Given that the first good remains unallocated, though, this similarly implies that whoever receives the second good would be envied, so both goods would have to remain unallocated, leading to a value of 0 for both agents.} To address this issue, the fair division literature has introduced the following approximate envy-freeness notions, adjusted to our setting:
\begin{definition}
    (\textbf{EF1}) An allocation $\alloc$ 
    is \emph{envy-free up to one good} (EF1) with respect to budgets if for any two agents $i, j \in N$ and every $S \subseteq \bundle_{j}$ with $c(S) \leq B_i$, there \emph{exists} a good $g \in S$ such that 
    $$v_i(\bundle_i) \geq v_i(S \setminus \set{g}).$$
\end{definition}

\begin{definition}
    (\textbf{EFx}) An allocation $\alloc$ 
    is \emph{envy-free up to any good} (EFx) with respect to budgets if for any two agents $i, j \in N$, every $S \subseteq \bundle_{j}$ with $c(S) \leq B_i$ and \emph{all} goods $g \in S$, we have 
    $$v_i(\bundle_i) \geq v_i(S \setminus \set{g}).$$
\end{definition}

We also use the term \emph{EFx-envy} to refer to a violation of the EFx property for some pair of agents, as follows:

\begin{definition}
    In allocation $\alloc$, an agent $i$ \emph{EFx-envies} another agent $j$ if there exists $S \subseteq \bundle_{j}$ with $c(S) \leq B_i$ and $g \in S$ such that 
    $$v_i(\bundle_i) < v_i(S \setminus \set{g}).$$
\end{definition}

Given these envy notions, we define the EFx-feasibility graph $\EFxGraph$ which indicates the bundles we could assign to each agent such that they would not EFx-envy any other agents.
\begin{definition}
    Given a subset of agents $N' \subseteq N$ and a set $\mathcal{S}$ of disjoint bundles of goods (i.e., $S_i\cap S_j=\emptyset$ for all $S_i,S_j\in \mathcal{S}$), the EFx-feasibility graph $\EFxGraph(N', \mathcal{S})$ is an undirected bipartite graph whose edges are defined as follows:
    $$E(\EFxGraph) = \set{(i, S_j)~:~\vmax_i(S_j) \geq \vmax_i(S_k \setminus \set{g})~\forall k,~\forall g \in S_k}$$
\end{definition}
In other words, if there is an edge between some agent $i$ and bundle $S_j$ in $\EFxGraph$, then allocating $S_j$ to agent $i$ will satisfy the EFx constraint for her, irrespective of how the other bundles are allocated. Therefore, a perfect matching in $\EFxGraph$ corresponds to an EFx allocation (and hence it is ``EFx-feasible'').

Note that it is always possible to remove any type of envy by keeping all the items unallocated. However, this comes at a great cost in terms of efficiency. To avoid this, we combine envy-freeness guarantees (which capture fairness) with (approximate) guarantees in terms of Pareto efficiency:
\begin{definition}
    An allocation $\alloc$ is Pareto-efficient if there does not exist any other allocation $\alloc'$ such that $v_i(\bundle_i') > v_i(\bundle_i)$ for some agent $i$, and $v_j(\bundle_j') \geq v_j(\bundle_j)$ for all agents $j\in N$.
\end{definition}

Our first result in Section~\ref{sec:2agents-instance} shows that even for an instance involving just three items and two agents with the same budget, it is impossible to simultaneously guarantee EFx (in fact, not even EF1) and Pareto efficiency. We, therefore, turn to approximation and, among the multiple Pareto-efficient allocations, we use as a benchmark the particularly appealing Pareto-efficient allocation that maximizes the Nash social welfare ($\NSW$) objective:
\begin{definition}
    The Nash social welfare of an allocation $\alloc$ is
    \[
        \NSW(\alloc) = \left(\prod_{i \in N} v_i(\bundle_i) \right)^{1/n}.
    \]
\end{definition}

Specifically, we use the term $\rho$-efficiency to determine how closely our outcome approximates the optimal Nash social welfare.

\begin{definition} If $\allocSet$ is the set of all budget-feasible allocations and $\OPT = \argmax_{\alloc \in \allocSet} \NSW(\alloc)$ is a maximum Nash welfare allocation, then $\alloc\in \allocSet$ is \emph{$\rho$-efficient}, i.e., a $\rho$ approximation, if 
$$\NSW(\alloc) \geq \rho \cdot \NSW(\OPT).$$ 
\end{definition}

Note that, since our benchmark, $\OPT$, is Pareto-efficient, any \emph{$\rho$-efficient} allocation is also approximately Pareto-efficient, with approximation factor $\rho$. Without loss of generality, we normalize the agents' valuations so that for each agent $i$, their value in the Nash welfare maximizing solution is equal to $1$, i.e., $v_i(\OPT_i) = 1$.\footnote{Note that since the $\NSW$ objective is scale-independent, scaling an agent's values for each item by the same constant does not affect the Nash welfare optimal outcome, just its value, so this is without loss of generality.}

Finally, in section \ref{sec:3agents-alg}, we use the term \emph{monopoly value} of some given agent $i$ with respect to a given budget $B$ to refer to the maximum value that this agent could achieve if she were allowed to choose \emph{any} bundle $S\subseteq M$ with cost at most $B$.
\begin{definition}
    For each agent $i$, let $m_i(B)$ denote the \emph{monopoly value of agent $i$ with respect to budget $B$}, i.e. 
    \[
        m_i(B) = \max_{S \subseteq M~:~c(S) \leq B} v_i(S).
    \]
\end{definition}

%% file: EF1andPO.tex
To exhibit the difficulties that arise in the presence of budgets, we first show that simultaneously guaranteeing EFx (in fact, even EF1) and Pareto-efficiency is infeasible. We exhibit this using an instance with just three items and two agents with equal budgets. Furthermore, using the same instance we prove an impossibility result regarding the best approximation of the maximum Nash welfare that is achievable by EF1 allocations.

\begin{theorem}\label{thm:lowerBound}
    There exists a problem instance with two agents of equal budget such that no budget-feasible allocation is both EF1 and Pareto-efficient. Furthermore, for the same instance, no budget-feasible allocation is both EF1 and $(\sqrt{1/2} + \varepsilon)$-efficient, for $\varepsilon > 0$.
\end{theorem}
\begin{proof}
Consider the following instance with two agents, where $B_1=B_2=1$, and the set of three items $\{1, 2, 3\}$, whose costs are:
    \[
        c(g) = \begin{cases}
            1/2, &\text{for } g \in \set{1,2}\\
            1, &\text{for } g = 3
        \end{cases}
    \]
    The values of the items from each agent's perspective are as follows:
    \[
        v_1(g) = \begin{cases}
            1/2, &g \in \set{1,2}\\
            0, &g = 3
        \end{cases}
        \hspace{1cm}
        v_2(g) = \begin{cases}
            1+\varepsilon, &g \in \set{1,2}\\
            1, &g = 3.
        \end{cases}
    \]
    For any $\varepsilon<1$, the budget-feasible allocation that maximizes $\NSW$ allocates items $\set{1,2}$ to agent 1 and item $\set{3}$ to agent 2, leading to a $\NSW$ of 1. To verify this, note that if agent 1 were to receive none of these two items her value, and thus also the $\NSW$, would be 0. Also, if agent 1 received just one of these two items, for a value of $1/2$, the value of agent 2 would be at most $1+\varepsilon$ (since she can afford just one of the remaining two items), leading to $\NSW$ less than 1. 
    
    However, the $\NSW$ maximizing allocation is not EF1 for agent 2, and any EF1 allocation for agent 2 has to leave one item unallocated and is not Pareto-efficient. The budget-feasible EF1 allocation with the largest $\NSW$ is $\alloc=(\set{1}, \set{2})$, with $\NSW(\alloc)=\sqrt{1/2 \cdot (1+\varepsilon)}< \sqrt{1/2} + \varepsilon$. Therefore, no budget-feasible EF1 allocation can achieve a $\NSW$ approximation of $\sqrt{1/2} + \varepsilon$.
\end{proof}

%% file: 2agents.tex
We now propose Procedure \ref{alg:2agents}, which takes as input a set of two agents (labeled 1 and 2) and an arbitrary budget-feasible allocation for these two agents, and returns a budget-feasible allocation that is EFx (therefore also EF1) with a $\NSW$ at least a $\sqrt{1/2}$ fraction of the original allocation's $\NSW$. 
Therefore, if we choose the $\NSW$ optimal allocation as the original allocation, this procedure returns a $\sqrt{1/2}$-efficient EFx allocation. Note that this is optimal in quite a strong sense, as it achieves the best approximation of the optimal $\NSW$ that is possible not just by EFx allocations, but even for the more permissive family of EF1 allocations, as shown in Theorem~\ref{thm:lowerBound}. 

Our procedure first checks whether the input allocation is already EFx, in which case it simply returns this allocation, or whether the agents both EFx-envy each other, in which case it ``swaps'' their bundles (while respecting their budget constraints) and terminates. Otherwise, if just one of them envies the other, for simplicity we reindex the agents so that it is agent 1 who EFx envies agent 2. The procedure then continues with two different approaches based on the ``amount" of envy agent 1 has towards agent 2's bundle, $\bundle_2$.

If agent 1 prefers $\bundle_2$ at most 2 times more than she likes her own, we proceed as follows: agent 1 repeatedly removes her least valued good $g$ from $\bundle_2$ and sets it aside in a separate bundle, $R$. The procedure terminates when it can find a ``matching'' between the two agents and two of the three bundles ($\bundle_1, \bundle_2$, or $R$) that yields an EFx allocation. The crucial observation is that agent 1 removes items from $\bundle_2$ only while it EFx-envies that bundle (otherwise an EFx allocation is reached), and since she always removes her least valued item from it, even if after some removal she becomes EFx-feasible with another bundle, she will remain EFx-feasible with $\bundle_2$. This ensures that at some point if no matching has been found already, at least two of the three bundles will be EFx-feasible for agent 1 (specifically $\bundle_1$ and $\bundle_2$ as shown in Lemma~\ref{lemma:matching-two-edges}), allowing us to find a matching by giving agent 2 the one she prefers. Note that the existence of this matching relies on the assumption that $v_1(\bundle_1) \geq \frac{1}{2} \vmax_1(\bundle_2)$.

If this assumption is not true, we use a different approach - the $\leximin$ procedure (Algorithm 2) of Plaut and Roughgarden \cite{PR2020}. Intuitively, this procedure splits a bundle of items between two agents such that an EFx allocation can be achieved when agents have general valuations (which includes our budget-feasible $\vmax$ function). We use this to split $\bundle_2$ into two parts such that agent 1 would be EFx-feasible with both if these bundles were the only ones participating in the matching. This similarly allows us to prove the existence of at least two edges for agent 1 in the EFx-feasibility graph, which implies the existence of a perfect matching.

In order to prove the desired efficiency guarantees, we show a stronger statement in terms of individual value guarantees. More specifically, we show that Procedure~\ref{alg:2agents} always finds a perfect matching in which agent 1 (the envying agent) gets a weakly higher value than in the input allocation, whereas agent 2 (the envied agent) gets at least half of her original value.

\begin{algorithm}[H]
\DontPrintSemicolon
\caption{(\twoalgoDB) EFx allocation for 2 agents} \label{alg:2agents}
\textbf{Input}: Set of two agents $\set{1,2}$; budget-feasible allocation $\alloc$ for these agents\\
    \lIf{$\alloc$ is EFx} {
        \Return $\alloc$
    }
    \lIf{both of the agents EFx-envy each other in $\alloc$} {
        \Return $(\vargmax_1(\bundle_2), \vargmax_2(\bundle_1))$
    }
    Reindex the agents so that agent 1 EFx-envies agent 2 \label{line:reindex}\\
    \If{$v_1(\bundle_1) \geq \frac{1}{2}\vmax_1(\bundle_2)$} {
        $R \gets \emptyset$\\
        \While{$\EFxGraph(\set{1,2}, (\bundle_1,\bundle_2,R))$ does not have a perfect matching} {
            $g \gets \argmin_{h \in \bundle_2} v_1(h)$\\
            $\bundle_2 \gets \bundle_2 \setminus \set{g}$\\
            $R \gets R \cup \set{g}$\\
            Update the edges of $\EFxGraph$
        }
    } \Else {
        $(\bundle_2', \bundle_2'') \gets \leximin(\bundle_2, \vmax_1, \vmax_1)$ \Comment*[r]{agent 1 splits $\bundle_2$ into two bundles}
        Construct $\EFxGraph(\set{1,2}, (\bundle_1, \bundle_2', \bundle_2''))$
    }
    Let $\set{(1, \bundle_1^*), (2, \bundle_2^*)}$ be a perfect matching in $\EFxGraph$ which\\
        \Indp (a) maximizes the value of agent 2 \label{cond:a}\\
        (b) under (a) maximizes the value of agent 1 \label{cond:b}\\ \Indm
    \Return $(\vargmax_1(\bundle_1^*), \vargmax_2(\bundle_2^*))$
\end{algorithm}

For the following lemmas, let $g_t$ be the item removed in iteration $t$ of the while loop for the case when $v_1(\bundle_1) \geq \frac{1}{2}\vmax_1(\bundle_2)$. Also, let $\bundle_2^t$ and $R^t$ be the state of the bundles at the end of this iteration (after moving $g_t$ from $\bundle_2^t$ to $R^t$). Note that if the input allocation $\alloc$ is already EFx, or if both agents EFx-envy each other in $\alloc$, we can immediately return a budget-feasible EFx allocation in which both agents (weakly) improve their value. If none of these two statements are true, it must be that one agent (agent 1 by the reindexing in line~\ref{line:reindex}) EFx-envies the other, while the second agent does not envy the first. For this reason, for the rest of the analysis, we assume that agent 1 EFx-envies agent 2.

Note that every agent has at least one edge in $\EFxGraph$ since they must be EFx-feasible with at least their favorite bundle.

\begin{lemma}\label{lemma:matching-two-edges}
    If any agent has at least two edges in $\EFxGraph$, a perfect matching can be found.
\end{lemma}
\begin{proof}
    Without loss, say agent 1 has at least two edges in $\EFxGraph$. Then, we can match agent 2 with her favorite bundle (which must be EFx-feasible to her) and agent 1 with one of her EFx-feasible bundles, besides the one matched to agent 2, which is guaranteed to exist since she is feasible with at least two bundles.
\end{proof}

\begin{lemma} \label{lemma:matching-exists}
    If $v_1(\bundle_1) \geq \frac{1}{2} \vmax_1(\bundle_2)$, Procedure~\ref{alg:2agents} always finds a perfect matching in $\EFxGraph$.
\end{lemma}
\begin{proof}
    If the input allocation forms a perfect matching, we are done. Otherwise, both agents are initially only feasible with a single bundle, by Lemma~\ref{lemma:matching-two-edges}, and since there is no perfect matching, both agents must only be feasible with the same bundle. By the reindexing in line~\ref{line:reindex}, this bundle is $\bundle_2$.
    
    Clearly, removing items from $\bundle_2$ will eventually create a new edge in $\EFxGraph$ for agent 1, either towards $\bundle_1$ or $R$. Since $R$ is initially the empty set and agent 1 moves one item from $\bundle_2$ to $R$ at each iteration, there exists an iteration $t$ such that $\vmax_1(\bundle_2^t) < \frac{1}{2}\vmax_1(\bundle_2)$, where $\bundle_2$ is the input bundle. Let $t$ be the first such iteration. If a perfect matching is found before iteration $t-1$, we are done.  Otherwise, in iteration $t-1$, we have $\vmax_1(\bundle_2^{t-1}) \geq \frac{1}{2}\vmax_1(\bundle_2)$. Since $\vmax_1(\bundle_2^t) = \vmax_1(\bundle_2^{t-1} \setminus \set{g_t}) < \frac{1}{2}\vmax_1(\bundle_2)$, it must be that in iteration $t-1$, agent 1 becomes EFx-feasible with $\bundle_1$ due to the assumption that $v_1(\bundle_1) \geq \frac{1}{2} \vmax_1(\bundle_2)$.

    As shown earlier in this proof, agent is initially EFx-feasible with only $\bundle_2$. By Lemma~\ref{lemma:matching-two-edges}, since a perfect matching was not found before iteration $t-1$, it must be that agent 1 had a single feasibility edge throughout these iterations, and so at each iteration $t' < t-1$, it must be that $\vmax_1(\bundle_2^{t'} \setminus \set{g_{t'}}) = \vmax_1(\bundle_2^{t'+1}) > v_1(\bundle_1)$ (note that $g$ is picked in such a way that this inequality holds). Thus, in iteration $t-1$, agent 1 is also EFx-feasible with $\bundle_2^{t-1}$.

    Since agent 1 is EFx-feasible with both $\bundle_1$ and $\bundle_2^{t-1}$ in iteration $t-1$, a perfect matching exists by Lemma~\ref{lemma:matching-two-edges}.
\end{proof}

\begin{lemma} \label{lemma:2agents1}
    If $v_1(\bundle_1) \geq \frac{1}{2}\vmax_1(\bundle_2)$, Procedure~\ref{alg:2agents} returns an EFx allocation such that agent 1 gets a weakly higher value than her initial value, while agent 2 gets at least half of her initial value, Moreover, no agent envies the unallocated bundle.
\end{lemma}
\begin{proof}
    Note that a perfect matching in $\EFxGraph$ implies an EFx allocation. If the input allocation forms a perfect matching, we are done. Otherwise, by Lemma~\ref{lemma:matching-exists}, a perfect matching always exists. Let $t-1$ be the final iteration of the while loop (before a matching is returned), and $\bundle_1^*$ and $\bundle_2^*$ be the bundles in the returned matching assigned to agents 1 and 2. Note that $\max\set{v_2(\bundle_2^{t-1}), v_2(R^{t-1})} \geq \frac{1}{2} v_2(\bundle_2)$, where $\bundle_2$ is the input bundle, since $\bundle_2^{t-1} \cup R^{t-1} = \bundle_2$. Since agent 2 is EFx-feasible with her favorite bundle and the returned matching maximizes agent 2's value, it must be that
    \[
        \vmax_2(\alloc_2^*) \geq \max\set{v_2(\bundle_1), v_2(\bundle_2^{t-1}), v_2(R^{t-1})} \geq \frac{1}{2} v_2(\bundle_2).
    \]
    Also, agent 2 clearly does not envy the unallocated bundle, since $\alloc_2^*$ is her favorite bundle.

    By the argument of Lemma~\ref{lemma:matching-exists}, agent 1 becomes EFx-feasible with $\bundle_1$ in iteration $t'-1$, where $t'$ is the earliest iteration for which $\vmax_1(\bundle_2^{t'}) < \frac{1}{2}\vmax_1(\bundle_2)$. Notice that $\vmax_1(\bundle_2^{t'-1}) > \vmax_1(R^{t'-1})$, since $\bundle_2^{t'-1} \cup R^{t'-1} = \bundle_2$. Therefore, agent 1 could not have been EFx-feasible with $R^{t'-2}$ in iteration $t'-2$. This implies that in the returned matching, agent 1 can be matched to either $\bundle_1$ or $\bundle_2^{t'-1} = \bundle_2^{t-1}$. If she is matched to $\bundle_1$, then clearly $\vmax_1(\bundle_1^*) = v_1(\bundle_1)$, and if she is matched to $\bundle_2$, then $\vmax_1(\bundle_1^*) = \vmax_1(\bundle_2^{t-1}) > v_1(\bundle_1)$ by the argument of Lemma~\ref{lemma:matching-exists}. Furthermore, agent 1 will similarly not envy the unallocated bundle since otherwise matching her to this bundle would increase her value and violate condition (b) in line~\ref{cond:b}.
\end{proof}

\begin{lemma} \label{lemma:agent1-feasibility}
    If $v_1(\bundle_1) < \frac{1}{2}\vmax_1(\bundle_2)$, Procedure~\ref{alg:2agents} always finds a perfect matching in $\EFxGraph$.
\end{lemma}
\begin{proof}
    If the input allocation forms a perfect matching, we are done. Otherwise, since $\bundle_2' \cup \bundle_2'' = \bundle_2$, it must be that $\max\set{\vmax_1(\bundle_2'), \vmax_1(\bundle_2'')} \geq \frac{1}{2} \vmax_1(\bundle_2) > v_1(\bundle_1)$. Since one of $\bundle_2'$ and $\bundle_2''$ is agent 1's favorite bundle (wlog $\bundle_2'$), and so it must be EFx-feasible. Note that $\vmax_1(\bundle_2'') \geq \vmax_2(\bundle_2' \setminus \set{g})$ for any $g \in \bundle_2'$ by the outcome of $\leximin$. If $v_1(\bundle_1) < \vmax_1(\bundle_2'')$, then agent 1 is also EFx-feasible with $\vmax_1(\bundle_2'')$. Otherwise, if $v_1(\bundle_1) \geq \vmax_1(\bundle_2'')$, then $v_1(\bundle_1) \geq \vmax_2(\bundle_2' \setminus \set{g})$ and so agent 1 is also EFx-feasible with $\bundle_1$. Since agent 1 is EFx-feasible with at least two bundles, a perfect matching exists.
\end{proof}

\begin{lemma} \label{lemma:2agents2}
    If $v_1(\bundle_1) < \frac{1}{2}\vmax_1(\bundle_2)$, Procedure~\ref{alg:2agents} returns an EFx allocation such that one agent gets a weakly higher value than her initial value, while the other gets at least half of her initial value, Moreover, no agent envies the unallocated bundle.
\end{lemma}
\begin{proof}
    Lemma~\ref{lemma:agent1-feasibility} implies the existence of a perfect matching and also that agent 1 is EFx-feasible with at least two bundles. Note that $\max\set{v_2(\bundle_2'), v_2(\bundle_2'')} \geq \frac{1}{2} v_2(\bundle_2)$, since $\bundle_2' \cup \bundle_2'' = \bundle_2$. Since agent 2 is EFx-feasible with her favorite bundle and the returned matching maximizes agent 2's value, it must be that
    \[
        \vmax_2(\bundle_2^*) \geq \max\set{v_2(\bundle_2'), v_2(\bundle_2'')} \geq \frac{1}{2} v_2(\bundle_2).
    \]
    Also, for the same reason, she will not envy the unallocated bundle.

    Lemma~\ref{lemma:agent1-feasibility} also implies that at least two of agent 1's feasible bundles will have value at least $v_1(\bundle_1)$ for her. Since the returned matching will allocate at least one of these two bundles to agent 1 (by the maximality of condition (b) in line~\ref{cond:b}) and so her value $\vmax_1(\bundle_1^*)$ for her allocated bundle is weakly better than her initial value. Moreover, agent 1 will not envy the unallocated bundle or otherwise, we could increase agent 1's value and contradict the maximality of condition (b) in line~\ref{cond:b}.
\end{proof}

We now show the main result of this section.

\begin{theorem} \label{thm:2agents}
    $\twoalgoDB$ returns a budget-feasible EFx allocation such that one agent gets a weakly higher value than her initial value, while the other gets at least half of her initial value, Moreover, no agent envies the unallocated bundle.
\end{theorem}
\begin{proof}
    Note that the returned allocation is budget-feasible by definition of $\vargmax$. Combining Lemmas~\ref{lemma:2agents1} and \ref{lemma:2agents2} gives the desired efficiency and envy guarantees.
\end{proof}

Note that since $\twoalgoDB$ returns a budget-feasible allocation, some items from the matched bundles, $\bundle_1^*$ and $\bundle_2^*$, might be left out (if the agent they are allocated to does not have sufficient budget for the whole bundle). In Lemma~\ref{lemma:left-out}, we show that this is the case with at most one of $\bundle_1^*$ and $\bundle_2^*$, and the left-out part is not EFx-envied (respectively envied) by the agents. We use this lemma for the three agents procedure later on.

\begin{lemma} \label{lemma:left-out}
    In the allocation returned by $\twoalgoDB$, one of $\bundle_1^*$ and $\bundle_2^*$ will be completely allocated (to the agent with the higher budget), while the other bundle may be partially allocated. If $R'$ is the left-out part, then one agent will not envy $R'$, while the other will not EFx-envy it.
\end{lemma}
\begin{proof}
    Note that the allocated bundles $\bundle_1^*$ and $\bundle_2^*$ are subsets of the input bundles, which are budget-feasible for the respective agents. Without loss of generality, let agent 2 be the one with the (weakly) higher budget. Thus, the input bundles $\bundle_1$ and $\bundle_2$ are budget-feasible for agent 2. Since $\bundle_2^*$ is a subset of one of these two bundles, it is budget-feasible for agent 2. Therefore, $\bundle_2^*$ is completely allocated and so $R'$ (possibly empty) must be a subset of $\bundle_1^*$.

    Since agent 2 does not EFx-envy $\bundle_1^*$ and $R'$ is a subset of this bundle, agent 2 will not EFx-envy $R'$. For agent 1, by definition, her most valuable budget-feasible subset of $\bundle_1^*$ is $\vargmax_1(\bundle_1^*)$. This implies that she cannot envy $R'$ (w.r.t. her budget) or otherwise $R'$ should be her most valuable budget-feasible subset of $\bundle_1^*$.
\end{proof}

%% file: 3agents.tex
In the case of two-agent instances, we were able to design a somewhat simple procedure in which the envying agent splits the envied bundle into two parts in such a way that an EFx allocation with optimal efficiency is returned. In the case of three-agent instances, however, things become significantly more complicated, since two agents may envy the third agent, yet disagree on how their bundle should be split. Furthermore, the fact that each of the three agents can have a vastly different budget, adds to the complexity of the problem. To address this issue, we first try to reduce the problem to the case of equal budgets. In order to do so, we consider the ``monopoly value'' of the higher-budget agents, if we were to limit their budget to the lowest one (see Section~\ref{sec:preliminaries} for a definition). If all agents' monopoly value with the smallest budget is ``high enough'' (this threshold is determined by a parameter $\alpha$ whose exact value we determine later on), then we reduce the problem to the case of equal budgets. If not, then we essentially let the agents arrive in increasing order of their budgets and choose their preferred budget-feasible bundle among the remaining goods. Note that we assume that $|M| > 3$, since otherwise, the problem is trivial - the maximum $\NSW$ allocation is EFx.


Our main procedure, Procedure \ref{alg:3agents}, first reindexes the agents in increasing order of their budgets, such that agent 1 has the smallest budget, followed by agent 2, and then 3. Then, it runs a pre-processing phase (described in Section~\ref{sec:preproc}) which allows each agent to set aside a single ``high value'' good. Next, our procedure checks whether we can reduce the input instance to one where all agents have equal budgets $B_1$, while still guaranteeing constant efficiency through the monopoly value $m(B_1)$, and if so, we call Procedure \ref{alg:3agents-equalbudget} which returns an EFx allocation for instances with equal budgets (see Section~\ref{sec:equal_budgets}). Otherwise, we run Procedure \ref{alg:3agents-else} (see Section~\ref{sec:3agents-else}). Once one of these subroutines returns an allocation, we provide each agent with the option of dropping the bundle they were assigned in exchange for the ``high-value'' good set aside for them during the pre-processing phase. This ensures that agents with a lot of value concentrated on a single item have the opportunity to achieve at least that much value if it is set aside for them in the pre-processing phase.

\begin{algorithm}[H]
\caption{(\threealgoDB) EFx allocation for 3 agents} \label{alg:3agents}
\textbf{Input}: Parameter $\threshold$\\
    Reindex the agents in increasing order of their budgets, such that $B_1 \leq B_2 \leq B_3$\\
    $M, (s_1, s_2, s_3) \gets$ Procedure \ref{alg:3agents-preprocess}\\
    \If{$m_2(B_1) \geq \threshold$ and $m_3(B_1) \geq \threshold$} {
        $B_2, B_3\gets B_1$ \label{line:reduce-budgets}\\
        $\alloc \gets$ Procedure \ref{alg:3agents-equalbudget}\label{line:run-3agents-equalbudget}
    } \Else {
        $\alloc \gets$ Procedure \ref{alg:3agents-else}
    }
    \For{$i \in N$} {
        \If{$v_i(s_i) > v_i(\bundle_i)$} {
            $\bundle_i \gets s_i$
        }
    }
    \Return $\alloc$
\end{algorithm}


\subsection{The Pre-Processing Step: Procedure \ref{alg:3agents-preprocess}}\label{sec:preproc}

We now present the pre-processing subroutine that we rely on to set aside one item for each agent. This subroutine tries to match them with one of their favorite three items and set it aside to ensure that they will never end up with a lower value. We carefully design this procedure so that this matching satisfies some useful properties that we use later on to prove the desired EFx and efficiency guarantees. In more detail, we use the maximum $\NSW$ allocation to ``prioritize'' the agents who have a favorite item in their respective maximum $\NSW$ bundle. Specifically, we guarantee that all these agents will be matched to a good that is at least as valuable as the most valuable good in their bundle (note that this can be achieved by just matching each such agent with her favorite good in her bundle). Then, given this hard constraint, we choose the maximum weight matching using the agents' values for the corresponding goods as the weights. This procedure returns the goods matched to each agent (observe that each agent gets one), sets them aside, and removes them from $M$. Note that Procedure \ref{alg:3agents}, at the very end, gives each agent the option to pick their set-aside item if that would increase their value, which acts as a safety net for their value guarantee.

\begin{algorithm}[H]
\caption{Pre-processing for 3 agents} \label{alg:3agents-preprocess}
\textbf{Input}: Set of items $M$, maximum $\NSW$ allocation $\OPT(M)$\\
    For each agent $i$, let $S_i$ be the set of their most favorite three budget-feasible goods in $M$\\
    Let $G$ be a bipartite graph with vertices $\set{1,2,3}$, corresponding to agents, on one side, and vertices $S_1\cup S_2\cup S_3$, corresponding to goods, on the other. An edge with weight $v_i(g)$ connects each agent $i$ to each good $g \in S_i$\\
    Let $\mathcal{P}$ be the set of all matchings of $G$ such that every agent $i$ with $S_i^* = S_i \cap \OPTBundle_i \neq \emptyset$ is matched to an item of value at least $\max_{g \in S_i^*} v_i(g)$\\
    Let $P$ be a maximum weight matching in $\mathcal{P}$ and for each agent $i$, let $\setaside_i$ be the good they are matched to in $P$\\
    $M \gets M \setminus \set{\setaside_1, \setaside_2, \setaside_3}$\\
    \Return $M, (\setaside_1, \setaside_2, \setaside_3)$
\end{algorithm}

\begin{lemma} \label{lemma:preprocess-maxval}
    At the end of Procedure \ref{alg:3agents-preprocess}, $v_i(\setaside_i) \geq v_i(g)$ for all agents $i$ and all $g \in M \setminus \set{\setaside_1, \setaside_2, \setaside_3}$.
\end{lemma}
\begin{proof}
    First, note that all agents are matched to an item since there are 3 agents and at least 3 items in $S_1 \cup S_2 \cup S_3$. Assume by contradiction that some agent $i$ is matched to an item $\setaside_i$ and $v_i(\setaside_i) < v_i(g)$ for some item $g \in M \setminus \set{s_1, s_2, s_3}$. Then, $g$ should either be in $S_i$ or agent $i$ should be matched to $g$ instead of $\setaside_i$ (which would increase the weight of the matching).
\end{proof}

From this point on, let $\OPT$ be the maximum $\NSW$ allocation on the remaining items $M \setminus \set{\setaside_1, \setaside_2, \setaside_3}$. We now lower bound each agent's value in the maximum $\NSW$ allocation on the remaining items using the following lemma.

\begin{lemma} \label{lemma:preprocess-remainingval}
    At the end of Procedure \ref{alg:3agents-preprocess}, at least one of the following three cases must hold for agents $i, j, k$:
    \begin{itemize}
        \item $v_i(\OPTBundle_i) \geq 1-3v_i(\setaside_i)$ and $v_j(\OPTBundle_j) = v_k(\OPTBundle_k) = 1$
        \item $v_i(\OPTBundle_i) \geq 1-2v_i(\setaside_i)$ and $v_j(\OPTBundle_j) \geq 1-v_j(\setaside_j)$ and $v_k(\OPTBundle_k) = 1$
        \item $v_i(\OPTBundle_i) \geq 1-v_i(\setaside_i)$ and $v_j(\OPTBundle_j) \geq 1-v_j(\setaside_j)$ and $v_k(\OPTBundle_k) \geq 1-v_k(\setaside_k)$
    \end{itemize}
\end{lemma}
\begin{proof}
    By the definition of the procedure, each agent $i$ is matched with an item of value at least $\max_{g \in S^*_i} v_i(g)$. If all three agents are matched to items in $\OPTBundle_i$, $i$'s value for these items is no more than $v_i(\setaside_i)$, so she can lose no more than $3v_i(\setaside_i)$ by removing the set-aside items, and her remaining value is at least $1-3v_i(\setaside_i)$. Since all the set-aside items were part of $\OPTBundle_i$ and no other items are removed, the other agents will have their initial maximum $\NSW$ value left among the remaining items.

    If two agents are matched to items in $\OPTBundle_i$ and the last agent is matched to an item in $v_j(\OPTBundle_j)$, $i$'s remaining value after removing these is at least $1-2v_i(\setaside_i)$, while $j$'s remaining value is at least $1-v_j(\setaside_j)$, similarly as above. The last agent $k$ loses no value.

    If all agents are each matched to an item in a different maximum $\NSW$ bundle, then the remaining value for each is at least 1 minus the value of their set-aside item.
\end{proof}


\subsection{The Equal Budget Case: Procedure~\ref{alg:3agents-equalbudget}}\label{sec:equal_budgets}

We now introduce Procedure \ref{alg:3agents-equalbudget} which outputs an EFx allocation with high $\NSW$ when the agents have equal budgets. The main idea behind this procedure is to reduce the budget-feasible problem to the case of unrestricted additive valuations. To do this, it removes a fraction of each agent's bundle in the $\NSW$ maximizing allocation, so that the set of all remaining items $Z$ is affordable for every agent, i.e., the total cost across these items is within the common budget $B_1$. It does so, however, while keeping enough value for every agent. The procedure then computes a complete EFx allocation (i.e., all items in $Z$ are allocated), which was proven to exist by \citet{CGM2020} - note that their existence proof is algorithmic, so we can use this approach to find a complete EFx allocation. Lastly, as an optimization, we allow agents to swap bundles with each other if they envy each other since the computed EFx allocation only guarantees the absence of EFx-envy. Note that before the procedure starts, we assume that the budgets are normalized to 1, which is without loss of generality.

We crucially use the fact that the value of every agent for any item is upper bounded by the value of their set-aside item to ensure that each agent can maintain roughly $1/n$ of their original value using at most a $1/n$ of their total budget. Note that our procedure can handle any number of agents, however, we use $n=3$ for most of the analysis.

\begin{algorithm}[H]
\caption{EFx allocation for 3 agents with equal budgets} \label{alg:3agents-equalbudget}
\textbf{Input}: Maximum $\NSW$ allocation $\alloc = \OPT(M)$\\
    \For{$i \in N$} {
        \While{$c(\bundle_i) > \frac{1}{n}$} { \label{line:whileloop}
            $\bundle_i \gets \bundle_i \setminus \set{\argmin\limits_{g \in \bundle_i} \frac{v_i(g)}{c(g)}}$
        }
    } \label{line:remove}
    $Z \gets \cup_{i \in N} \bundle_i$\\
    $\allocALG \gets$ a complete EFx allocation of $Z$ \label{line:EFx-findalloc}\\
    \If{$i$ envies $j$ and $j$ envies $k$ and $k$ envies $i$ for any agents $i,j,k$} { \label{line:if-envy}
        $\allocALGBundle_i, \allocALGBundle_j, \allocALGBundle_k \gets \allocALGBundle_j, \allocALGBundle_k, \allocALGBundle_i$
    }
    \While{$i$ envies $j$ for any agents $i,j$} { \label{line:while-envy}
        $\allocALGBundle_i, \allocALGBundle_j \gets \allocALGBundle_j, \allocALGBundle_i$
    }
    \Return $\allocALG$
\end{algorithm}

\begin{lemma} \label{lemma:remaining-value}
    After line \ref{line:remove} of Procedure \ref{alg:3agents-equalbudget}, $v_i(\bundle_i) \geq \frac{v_i(\OPTBundle_i)}{n} - v_i(\setaside_i)$ for all agents $i \in \set{1,\ldots,n}$.
\end{lemma}
\begin{proof}    
    We can find an upper bound for the total value that the procedure removes from each bundle, in two ways. For each bundle $\bundle_i$ in the max $\NSW$ allocation, let $R_i$ be the set of items that we remove in the while loop on line \ref{line:whileloop}. Define $\rho_g$ as the density of each item $g$. Moreover, let $\ell_i$ be the last item removed from $\bundle_i$ in the while loop. Note that $\rho_{\ell_i}$ is upper bounded by $\rho_i^* = \dfrac{v_i(\setaside_i)}{c_{\ell_i}}$ for all agents $i$, by Lemma \ref{lemma:preprocess-maxval}. Thus, we get the following inequality:
    
    \begin{align}
        \sum_{g \in R_i} v_i(g) &= \sum_{g \in R_i} c(g) \cdot \rho_g \nonumber\\
         &\leq \sum_{g \in R_i} c(g) \cdot \rho_i^* \nonumber\\ 
         &\leq \left(1 - \frac{1}{n} + c_{\ell_i} \right) \cdot \rho_i^* \nonumber\\
         &= \left(1 - \frac{1}{n} + c_{\ell_i} \right) \cdot \frac{v_i(\setaside_i)}{c_{\ell_i}} \nonumber\\
         &= \left(1 - \frac{1}{n}\right) \cdot \dfrac{v_i(\setaside_i)}{c_{\ell_i}} + v_i(\setaside_i). \label{eq:inequality1}
    \end{align}

    Notice that when $c_{\ell_i}$ is infinitesimally small, this upper bound makes the sum of the values of the items removed to be unbounded. However, by value normalization, we also know that
    \[
        \sum_{g \in R_i} v_i(g) + \sum_{g \in \bundle_i \setminus R_i} v_i(g) = v_i(\OPTBundle_i),
    \]
    so we get the second upper bound:
    \begin{align}
        \sum_{g \in R_i} v_i(g) &= v_i(\OPTBundle_i)-\sum_{g \in \bundle_i \setminus R_i} v_i(g) \nonumber\\
         &\leq v_i(\OPTBundle_i)-\sum_{g \in \bundle_i \setminus R_i} c(g) \cdot \rho_i^* \nonumber\\
         &\leq v_i(\OPTBundle_i)-\rho_i^* \left( \dfrac{1}{n} - c_{\ell_i} \right) \nonumber\\
         &= v_i(\OPTBundle_i)-\dfrac{v_i(\setaside_i)}{c_{\ell_i}} \left(\frac{1}{n} - c_{\ell_i} \right) \nonumber\\
         &= v_i(\OPTBundle_i)-\frac{1}{n} \cdot \dfrac{v_i(\setaside_i)}{c_{\ell_i}} + v_i(\setaside_i). \label{eq:inequality2}
    \end{align}
    Thus, we get that
    \begin{align*}
        \sum_{g \in R_i} v_i(g) &\leq \min \set{(\ref{eq:inequality1}), (\ref{eq:inequality2})}\\
        &=\begin{cases}
             v_i(\OPTBundle_i)-\frac{1}{n} \cdot \dfrac{v_i(\setaside_i)}{c_{\ell_i}} + v_i(\setaside_i),\\
             \text{if } v_i(\setaside_i) \geq v_i(\OPTBundle_i) \cdot c_{\ell_i} \vspace{1em}\\
             \left(1 - \frac{1}{n}\right) \cdot \dfrac{v_i(\setaside_i)}{c_{\ell_i}} + v_i(\setaside_i),\\
             \text{if } v_i(\setaside_i) < v_i(\OPTBundle_i) \cdot c_{\ell_i}
        \end{cases}\\
        &\leq v_i(\OPTBundle_i) - \frac{v_i(\OPTBundle_i)}{n} + v_i(\setaside_i).
    \end{align*}
    Therefore, for all agents $i$, at the end of the while loop on line \ref{line:whileloop},
    \begin{align*}
        v_i(\bundle_i) &\geq v_i(\OPTBundle_i) - \left(v_i(\OPTBundle_i) - \frac{v_i(\OPTBundle_i)}{n} + v_i(\setaside_i)\right)\\
        &= \frac{v_i(\OPTBundle_i)}{n}-v_i(\setaside_i). \qedhere
    \end{align*}
\end{proof}

We now show that our budget-feasible valuation function $\vmax$ is subadditive, which we use in the later proofs.

\begin{lemma} \label{lemma:subadditive}
    For any agent $i$ with budget $B_i$, any bundle $\bundle$, and any item $g \in \bundle$, the following inequality holds
    \[
        \vmax_i(\bundle \setminus \set{g}) \geq \vmax_i(\bundle) - \vmax_i(g).
    \]
\end{lemma}
\begin{proof}
    Let $S = \vargmax_i(\bundle)$ and note that $\vmax_i(\bundle) = v_i(S)$. If $g \in \vargmax_i(\bundle)$, then $\vmax_i(g) = v_i(g)$ and so
    \[
        \vmax_i(\bundle) - \vmax_i(g) = v_i(S) - v_i(g) = v_i(S \setminus \set{g}) \leq \vmax_i(\bundle \setminus \set{g}),
    \]
    since $S \subseteq \bundle$.

    If $g \notin \vargmax_i(\bundle)$, then $\vmax_i(\bundle) = \vmax_i(\bundle \setminus \set{g})$, so $\vmax_i(\bundle \setminus \set{g}) \geq \vmax_i(\bundle) - \vmax_i(g)$, since $\vmax_i(g) \geq 0$.
\end{proof}

\begin{restatable}{rLem}{lemmaThreeAgentsEqual} \label{lemma:3agents-equal}
    For three-agent instances with equal budgets, Procedure \ref{alg:3agents-equalbudget} returns a budget-feasible EFx allocation with the following individual guarantees:
    \begin{itemize}
        \item $v_i(\allocALGBundle_i) \geq \dfrac{v_i(\OPTBundle_i)}{9} - \dfrac{v_i(\setaside_i)}{3}$ for an agent $i$
        \item $v_j(\allocALGBundle_j) \geq \dfrac{v_j(\OPTBundle_j)}{9} - \dfrac{2v_j(\setaside_j)}{3}$ for an agent $j \neq i$
        \item $v_k(\allocALGBundle_k) \geq \dfrac{v_k(\OPTBundle_k)}{9} - v_k(\setaside_k)$ for an agent $k \neq i,j$
    \end{itemize}
\end{restatable}
\begin{proof}
    By construction, $c(Z) \leq 1$ since $n$ agents contribute a cost of at most $1/n$ each. Notice that in the returned allocation, each agent gets a subset of $Z$, and so the final allocation is budget-feasible. The returned allocation is EFx since $\allocALG$ is an EFx allocation in line \ref{line:EFx-findalloc} and the while loop in line \ref{line:while-envy} only removes envy.

    We now aim to bound how much value each agent has left in $\allocALG$. Since $c(Z) \leq 1$, we can apply the result of \citet{CGM2020} and get $\allocALG$, a complete EFx allocation of the items in $Z$. Thus, the individual value for each agent in $\allocALG$ is guaranteed by the EFx-ness of the allocation. Notice that although this allocation is EFx, it might be possible that all three agents form an envy cycle (i.e., $i$ envies $j$, $j$ envies $k$, $k$ envies $i$ for any agents $i,j,k$) or two of them envy each other. The if statement in line \ref{line:if-envy} and the while loop in line \ref{line:while-envy} address this and swap the bundles of such agents until there are no envy cycles. This implies that at least one agent does not envy anyone, while another agent may only envy a single other agent. Specifically, for agents $i,j,k$, if $i$ does not envy anyone, then $j$ can only envy $i$, while $k$ may envy both $i$ and $j$. Thus, for such agent $i$ who does not envy anyone else, we can guarantee the following three inequalities:
    \begin{itemize}
        \item $v_i(\allocALGBundle_i) \geq \vmax_i(\allocALGBundle_j)$
        \item $v_i(\allocALGBundle_i) \geq \vmax_i(\allocALGBundle_k)$
        \item $v_i(\bundle_i) \leq v_i(\allocALGBundle_i) + \vmax_i(\allocALGBundle_j) + \vmax_i(\allocALGBundle_k)$,
    \end{itemize}
    which implies $v_i(\allocALGBundle_i) \geq \dfrac{v_i(\bundle_i)}{3}$.

    For agent $j$ who may only envy $i$, we have the following three inequalities:
    \begin{itemize}
        \item $v_j(\allocALGBundle_j) \geq \vmax_j(\allocALGBundle_i \setminus \set{g}) \quad \forall g \in \allocALGBundle_i$
        \item $v_j(\allocALGBundle_j) \geq \vmax_j(\allocALGBundle_k)$
        \item $v_j(\bundle_j) \leq \vmax_j(\allocALGBundle_i) + v_j(\allocALGBundle_j) + \vmax_j(\allocALGBundle_k)$.
    \end{itemize}
    Since $\vmax_j(g) \leq v_j(\setaside_j)$ by Lemma \ref{lemma:preprocess-maxval} and $\vmax_j(\allocALGBundle_i \setminus \set{g}) \geq \vmax_j(\allocALGBundle_i) - \vmax_j(g)$  by Lemma \ref{lemma:subadditive} for all $g \in \allocALGBundle_i$, we get that $v_j(\allocALGBundle_j) \geq \dfrac{v_j(\bundle_j) - v_j(\setaside_j)}{3}$.

    Lastly, for agent $k$, we similarly have the following three guarantees:
    \begin{itemize}
        \item $v_k(\allocALGBundle_k) \geq \vmax_k(\allocALGBundle_i \setminus \set{g_1}) \quad \forall g_1 \in \allocALGBundle_i$
        \item $v_k(\allocALGBundle_k) \geq \vmax_k(\allocALGBundle_j \setminus \set{g_2}) \quad \forall g_2 \in \allocALGBundle_j$
        \item $v_k(\bundle_k) \leq \vmax_k(\allocALGBundle_i) + \vmax_k(\allocALGBundle_j) + v_k(\allocALGBundle_k)$,
    \end{itemize}
    Since $\vmax_k(g_1) \leq v_k(\setaside_k)$ and similarly for $g_2$ by Lemma \ref{lemma:preprocess-maxval} and $\vmax_k(\allocALGBundle_i \setminus \set{g_1}) \geq \vmax_k(\allocALGBundle_i) - \vmax_k(g_1)$ by Lemma \ref{lemma:subadditive} and similarly for the following inequality, we get that $v_k(\allocALGBundle_k) \geq \dfrac{v_k(\bundle_k) - 2v_k(\setaside_k)}{3}$.
    
    By Lemma \ref{lemma:remaining-value}, $v_i(\bundle_i) \geq \frac{v_i(\OPTBundle_i)}{3} - v_i(\setaside_i)$ for all agents $i$, and so the statement holds.
\end{proof}

\begin{restatable}{rLem}{lemmaEqualBudget} \label{lemma:equalbudget}
     Procedure \ref{alg:3agents-equalbudget} returns an allocation $\allocALG$ that is budget-feasible, EFx, and $\sqrt[3]{\frac{\threshold^2}{15 \cdot 15 \cdot 18}}$-efficient.
\end{restatable}
\begin{proof}
    Note that Procedure \ref{alg:3agents-equalbudget} returns a budget-feasible EFx allocation by Lemma~\ref{lemma:3agents-equal}. We now look at the agents' value guarantees. By Lemma~\ref{lemma:3agents-equal}, one agent $i$ gets a value of at least $\frac{v_i(\OPTBundle_i)}{9} - \frac{v_i(\setaside_i)}{3}$, another agent $j$ gets at least $\frac{v_j(\OPTBundle_j)}{9} - \frac{2v_j(\setaside_j)}{3}$, while the third agent $k$ gets at least $\frac{v_k(\OPTBundle_k)}{9} - v_k(\setaside_k)$. Since each agent $i$ can choose her set-aside item $\setaside_i$ at the end, it must be that $v_i(\bundle_i) \geq \max(v_i(\allocALGBundle_i), v_i(\setaside_i))$. By Lemma \ref{lemma:preprocess-remainingval}, $v_i(\OPTBundle_i)$ can take three different values based on which bundle the set-aside items were initially part of and later removed from. Note that since line \ref{line:reduce-budgets} in $\threealgoDB$ reduces the budgets of agents 2 and 3 to $B_1$, we can only guarantee that $v_2(\OPTBundle_2) \geq \threshold$ and $v_3(\OPTBundle_3) \geq \threshold$ through their monopoly values. Via an exhaustive case analysis deferred to the appendix, we show that $\NSW(\allocALG) \geq \sqrt[3]{\frac{\threshold^2}{15 \cdot 15 \cdot 18}}$.
\end{proof}


\subsection{The Remaining Case: Procedure~\ref{alg:3agents-else}} \label{sec:3agents-else}

Below, we present the description of Procedure \ref{alg:3agents-else}, which we call when some agent cannot be satisfied with a reduced budget. This procedure starts by assigning agent 1 her monopoly bundle and running $\twoalgoDB$ on agents 2 and 3 with the remaining items. Intuitively, the subroutine checks how much more value agents 2 and 3 have for agent 1's bundle, handling the cases where at least one of them would lose a lot of value if their budget was reduced to $B_1$. If agents 2 and 3 do not envy agent 1, then we return the allocation. The case when one agent envies agent 1 and the other is ``far'' away from envying agent 1 (that will be agent 3 or otherwise we reindex the agents without loss) is more tricky and requires more careful analysis. At a high level, since agent 3 has much more value for her bundle compared to any other bundle she could get with a budget of $B_1$, she can afford to lose a few subsets of her bundle that cost at most $B_1$. Note that both agents 1 and 2 split agent 1's monopoly bundle, which fits into a budget of $B_1$, so as long as we guarantee that agent 3 gets a bundle that she values more than anything she could get with a budget of $B_1$, she will not envy the other agents.

In order to guarantee that the final allocation is EFx, we use the round-robin procedure in line \ref{line:round-robin} on $\bundle_3$ (agent 3's bundle). Intuitively, this procedure splits $\bundle_3$ into two parts, in a round-robin fashion, such that the resulting parts, $\bundle_3'$ and $\bundle_3''$, are ``roughly'' equal in value (up to one item) for agent 3.

\begin{algorithm}[H]
\caption{Subroutine for Procedure \ref{alg:3agents}} \label{alg:3agents-else}
\textbf{Input}: Set of items $M$\\
    $\bundle_1 \gets \vargmax_1(M)$\\
    $(\bar{\bundle}_2,\bar{\bundle}_3) \gets$ max $\NSW$ allocation for agents 2, 3 and items $M \setminus \bundle_1$ \label{line:maxNSW23}\\
    $(\bundle_2, \bundle_3)\gets \twoalgoDB(\set{2,3}, (\bar{\bundle}_2,\bar{\bundle}_3))$ \label{line:2agentalgo-23}\\
    \If{$m_2(B_1) < \threshold$ and $m_3(B_1) < \threshold$} { \label{line:if-1}
        \Return $(\bundle_1, \bundle_2, \bundle_3)$ \label{line:return1-else}
    }
    Reindex agents 2 and 3 such that $m_3(B_1) < \threshold$ \label{line:swap-agents}\\
    \If{$v_2(\bundle_2) \geq v_2(\bundle_1)$} { \label{line:if-2}
        \Return $(\bundle_1, \bundle_2, \bundle_3)$ \label{line:return2-else}
    }
    $(\bundle_1', \bundle_2') \gets \twoalgoDB(\set{1,2}, (\emptyset, \bundle_1))$ \label{line:2agentalgo-12}\\
    Let agent 3 run round-robin by value with himself on $\bundle_3$ and $(\bundle_3', \bundle_3'')$ be the two resulting partitions\label{line:round-robin}\\
    $\bundle_3^* \gets \bundle_3 \setminus \argmax_{S \in \set{\bundle_3', \bundle_3''}} v_2(S)$ \label{line:x3*def}\\
    $\bundle_1'' \gets \vargmax_1(\bundle_3^*)$\\
    $\bundle_1^* \gets \bundle_1'$\\
    \If{$v_1(\bundle_1') < v_1(\bundle_1'')$} {
        $\bundle_1^* \gets \bundle_1''$ \label{line:agent1-pickbest}\\
        $\bundle_3^* \gets \bundle_3^* \setminus \bundle_1''$ \label{line:agent1-remove}
    }
    \Return $(\bundle_1^*, \bundle_2', \bundle_3^*)$ \label{line:return3-else}
\end{algorithm}

\begin{restatable}{rLem}{lemmaCaseOne} \label{lemma:return1-else}
    If Procedure \ref{alg:3agents-else} returns the allocation $\allocALG = (\bundle_1, \bundle_2, \bundle_3)$ in line \ref{line:return1-else}, this allocation is budget-feasible, EFx for any $\threshold \leq \frac{1}{10}$, and $\sqrt[3]{\frac{(1-\threshold)^2}{4 \cdot 6 \cdot 6}}$-efficient.
\end{restatable}
\begin{proof}
    First, notice that $\bundle_1 = \vargmax_1(M)$ is budget-feasible for agent 1 by definition. Since $\bundle_2$ and $\bundle_3$ result from running the $\twoalgoDB$ procedure, they are budget-feasible for agents 2 and 3 by Theorem \ref{thm:2agents} and so the returned allocation is budget-feasible.
    
    We now look at each agent's value guarantee. Note that agent 1 gets her monopoly bundle, so $v_1(\bundle_1) \geq v_1(\OPTBundle_1)$. In line~\ref{line:2agentalgo-23} agents 2 and 3 run $\twoalgoDB((\bar{\bundle}_2,\bar{\bundle}_3), M \setminus (\bundle_1 \cup \bar{\bundle}_2 \cup \bar{\bundle}_3))$, resulting in an EFx allocation $(\bundle_2, \bundle_3)$ for agents 2 and 3. We note that neither agent 2 nor 3 envies $M\setminus(\bundle_1 \cup \bar{\bundle}_2 \cup \bar{\bundle}_3)$ or otherwise they can get the envied subset resulting in an allocation with higher $\NSW$, contradicting the optimality of the allocation in line~\ref{line:maxNSW23}. Let $R$ be the set of unallocated items and $R'$ be the set of left-out items (as defined in Lemma~\ref{lemma:left-out}) after running $\twoalgoDB$ in line~\ref{line:2agentalgo-23}. We then get the following five inequalities:
    \begin{enumerate}
        \item $v_2(\bundle_2) \geq \vmax_2(\bundle_3 \setminus \set{g}) \quad \forall g \in \bundle_3$ \label{eq:agent2val-1}
        \item $v_2(\bundle_2) \geq \vmax_2(R)$ \label{eq:agent2val-2}
        \item $v_2(\bundle_2) \geq \vmax_2(R') \quad$ if $R' \subset \bundle_2$ \label{eq:agent2val-leftout}
        \item $v_2(\bundle_2) > \vmax_2(M\setminus(\bundle_1 \cup \bar{\bundle}_2 \cup \bar{\bundle}_3))$
        \item $v_2(\OPTBundle_2) \leq v_2(\bundle_1) + v_2(\bundle_2) + \vmax_2(\bundle_3) + \vmax_2(R) + \vmax_2(R') + \vmax_2(M\setminus(\bundle_1 \cup \bar{\bundle}_2 \cup \bar{\bundle}_3))$ \label{eq:agent2val-3}.
    \end{enumerate}
    Inequalities~(\ref{eq:agent2val-1}) and (\ref{eq:agent2val-2}) follow from the properties of $\twoalgoDB$, and inequality~(\ref{eq:agent2val-leftout}) follows from Lemma~\ref{lemma:left-out}. Inequality (\ref{eq:agent2val-3}) is true since the union of the bundles in the right-hand side is $M$. Since $v_2(\bundle_1) < \threshold$ by case assumption ($m_2(B_1) < \threshold$ by the if statement in line \ref{line:if-1}), $\vmax_2(g) \leq v_2(\setaside_2)$ by Lemma \ref{lemma:preprocess-maxval} and $\vmax_2(\bundle_3 \setminus \set{g}) \geq \vmax_2(\bundle_3) - \vmax_2(g)$ for all $g \in \bundle_3$ by Lemma \ref{lemma:subadditive}, we get that $v_2(\bundle_2) > \frac{v_2(\OPTBundle_2)-v_2(\setaside_2)-\threshold}{5}$.

    Similarly, for agent 3, we get the following five inequalities:
    \begin{itemize}
        \item $v_3(\bundle_3) \geq \vmax_3(\bundle_2 \setminus \set{g}) \quad \forall g \in \bundle_2$
        \item $v_3(\bundle_3) \geq \vmax_3(R)$
        \item $v_3(\bundle_3) \geq \vmax_3(R') \quad$ if $R' \subset \bundle_3$
        \item $v_3(\bundle_3) > \vmax_3(M\setminus(\bundle_1 \cup \bar{\bundle}_2 \cup \bar{\bundle}_3))$
        \item $v_3(\OPTBundle_3) \leq v_3(\bundle_1) + \vmax_3(\bundle_2) + v_3(\bundle_3) + \vmax_3(R) + \vmax_3(M\setminus(\bundle_1 \cup \bar{\bundle}_2 \cup \bar{\bundle}_3))$,
    \end{itemize}
    Thus, $v_3(\bundle_3) > \frac{v_3(\OPTBundle_3)-v_3(\setaside_3)-\threshold}{5}$, similarly as above and by case assumption which implies $v_3(\bundle_1) < \threshold$. 

    We obtain that the $\NSW$ guarantee is the following:
    \begin{align*}
        \NSW(\allocALG)^3 \geq& \max\left(v_1(\OPTBundle_1), v_1(\setaside_1)\right) \cdot\\
        &\max\left(\frac{v_2(\OPTBundle_2)-v_2(\setaside_2)-\threshold}{5}, v_2(\setaside_2)\right) \cdot\\
        &\max\left(\frac{v_3(\OPTBundle_3)-v_3(\setaside_3)-\threshold}{5}, v_3(\setaside_3)\right).
    \end{align*}
    
    Note that at least one of agents 2 and 3 will strictly not envy the other after running $\twoalgoDB$, however, to simplify the analysis, we just use the fact that they do not EFx-envy each other. Also, $R'$ can only be a subset of one of $\bundle_2$ and $\bundle_3$, whereas for simplicity we include both cases simultaneously. Via an exhaustive case analysis deferred to the appendix, we show that $\NSW(\allocALG) \geq \sqrt[3]{\frac{(1-\threshold)^2}{4 \cdot 6 \cdot 6}}$.

    For the envy, note that agent 1 does not envy agents 2 or 3 because she picked her monopoly bundle with respect to her budget. Agents 2 and 3 do not EFx envy each other by the property of $\twoalgoDB$. For $\threshold \leq \frac{1-\threshold}{9}$, they do not envy agent 1 since their value for $\bundle_1$ is at most $\threshold$ and they each get at least $\frac{1-\threshold}{9}$ (in the case where all the set-aside items were initially part of their maximum $\NSW$ bundle - see the complete proof in the appendix).
\end{proof}

\begin{restatable}{rLem}{lemmaCaseTwo} \label{lemma:return2-else}
    If Procedure \ref{alg:3agents-else} returns the allocation $\allocALG = (\bundle_1, \bundle_2, \bundle_3)$ in line \ref{line:return2-else}, this allocation is budget-feasible, EFx for any $\threshold \leq \frac{1}{10}$, and $\sqrt[3]{\frac{1-\threshold}{4 \cdot 7 \cdot 6}}$-efficient.
\end{restatable}
\begin{proof}
   Note that after line \ref{line:swap-agents}, the algorithm guarantees that $m_2(B_1) \geq \threshold$ and $m_3(B_1) < \threshold$. Since we do not perform any operations on the bundles as compared to their state in line \ref{line:return1-else}, the budget-feasibility and EFx guarantees for all agents remain the same as in Lemma \ref{lemma:return1-else}. Moreover, the individual value guarantees for agents 1 and 3 are the same as before. Only the value guarantees for agent 2 change, since now $v_2(\bundle_2) \geq v_2(\bundle_1)$. For agent 2's value, we get the following six inequalities:
    \begin{itemize}
        \item $v_2(\bundle_2) \geq v_2(\bundle_1)$
        \item $v_2(\bundle_2) \geq \vmax_2(\bundle_3 \setminus \set{g}) \quad \forall g \in \bundle_3$
        \item $v_2(\bundle_2) \geq \vmax_2(R)$
        \item $v_2(\bundle_2) \geq \vmax_2(R') \quad$ if $R' \subset \bundle_2$
        \item $v_2(\bundle_2) > \vmax_2(M\setminus(\bundle_1 \cup \bar{\bundle}_2 \cup \bar{\bundle}_3))$
        \item $v_2(\OPTBundle_2) \leq v_2(\bundle_1) + v_2(\bundle_2) + \vmax_2(\bundle_3) + \vmax_2(R) + \vmax_2(R') +  \vmax_2(M\setminus(\bundle_1 \cup \bar{\bundle}_2 \cup \bar{\bundle}_3))$
    \end{itemize}
    Similarly as in the previous proof, we get that $v_2(\bundle_2) \geq \frac{v_2(\OPTBundle_2)-v_2(\setaside_2)}{6}$.
    
    We obtain that the $\NSW$ guarantee is the following:
    \begin{align*}
        \NSW(\allocALG)^3 \geq& \max\left(v_1(\OPTBundle_1), v_1(\setaside_1)\right) \cdot\\
        &\max\left(\frac{v_2(\OPTBundle_2)-v_2(\setaside_2)}{6}, v_2(\setaside_2)\right) \cdot\\
        &\max\left(\frac{v_3(\OPTBundle_3)-v_3(\setaside_3)-\threshold}{5}, v_3(\setaside_3)\right).
    \end{align*}
    Via an exhaustive case analysis deferred to the appendix, we get that $\NSW(\allocALG) \geq \sqrt[3]{\frac{1-\threshold}{4 \cdot 7 \cdot 6}}$.
\end{proof}

\begin{restatable}{rLem}{lemmaCaseThree} \label{lemma:return3-else}
    If Procedure \ref{alg:3agents-else} returns the allocation $\allocALG = (\bundle_1^*, \bundle_2', \bundle_3^*)$ in line \ref{line:return3-else}, this allocation is budget-feasible, EFx for $\threshold \leq \frac{1}{35}$, and $\sqrt[3]{\frac{1-9\threshold}{6 \cdot 13 \cdot 21}}$-efficient.
\end{restatable}
\begin{proof}
    For budget-feasibility, note that $\bundle_1^*$ is budget-feasible for agent 1 since both $\bundle_1'$ and $\bundle_1''$ are budget-feasible for her. Agent 2 gets $\bundle_2' \subseteq \bundle_1$ by the outcome of $\twoalgoDB$ in line~\ref{line:2agentalgo-12}, and since $c(\bundle_1) \leq B_1 \leq B_2$, $\bundle_2'$ is budget-feasible for agent 2. Agent 3 also gets a budget-feasible bundle in $\bundle_3$ by the outcome of $\twoalgoDB$ in line~\ref{line:2agentalgo-23}, and since $\bundle_3^* \subseteq \bundle_3$, the returned allocation is budget-feasible.
    
    We now consider the individual value guarantees of the agents. Note that $\bundle_1$ is agent 1's monopoly bundle. By running $\twoalgoDB$ in line \ref{line:2agentalgo-12}, we get that $v_1(\bundle_1') \geq v_1(\bundle_2' \setminus \set{g})$ for all $g \in \bundle_2'$, and $v_1(R_1) = 0$, where $R_1$ is the set of unallocated items after running $\twoalgoDB$ here. Since $v_1(\bundle_1') + v_1(\bundle_2') = v_1(\bundle_1) \geq v_1(\OPTBundle_1)$ and agent 1 has the option to increase her value by picking $\bundle_1''$, we get $v_1(\bundle_1^*) \geq \frac{v_1(\OPTBundle_1) - v_1(\setaside_1)}{2}$. Note that there are no left-out items (as defined by Lemma~\ref{lemma:left-out}).

    For agent 2, we have $v_2(\bundle_2) \geq \vmax_2(\bundle_3 \setminus \set{g})$ for all $g \in \bundle_3$ from running $\twoalgoDB$ on agents 2 and 3 in line \ref{line:2agentalgo-23}. Since $v_2(\bundle_2) < v_2(\bundle_1)$ by case assumption and $v_2(\OPTBundle_2) \leq v_2(\bundle_1) + v_2(\bundle_2) + \vmax_2(\bundle_3) + \vmax_2(R_2) + \vmax_2(R_2') + \vmax_2(M\setminus(\bundle_1 \cup \bar{\bundle}_2 \cup \bar{\bundle}_3))$, where $R_2$ is the set of unallocated items and $R_2'$ is the set of left-out items after running $\twoalgoDB$ in line~\ref{line:2agentalgo-23}, we have $v_2(\bundle_1) > \frac{v_2(\OPTBundle_2)-v_2(\setaside_2)}{6}$; this is shown in more detail in the previous lemma. By the outcome of $\twoalgoDB$ in line \ref{line:2agentalgo-12}, we get $v_2(\bundle_2') \geq \frac{\vmax_2(\bundle_1)}{2} > \frac{v_2(\OPTBundle_2)-v_2(\setaside_2)}{12}$.
    
    Now consider the value of agent 3. Similarly, we know that $v_3(\OPTBundle_3) \leq v_3(\bundle_1) + \vmax_3(\bundle_2) + v_3(\bundle_3) + \vmax_3(R_2) + \vmax_3(R_2') + \vmax_3(M\setminus(\bundle_1 \cup \bar{\bundle}_2 \cup \bar{\bundle}_3))$, and $v_3(\bundle_1) < \threshold$ by case assumption, and $v_3(\bundle_3) \geq \vmax_3(\bundle_2 \setminus \set{g})$, for all $g \in \bundle_2$, and $v_3(\bundle_3) \geq \vmax_3(R_2)$ by the properties of $\twoalgoDB$. Thus, we get that $v_3(\bundle_3) > \frac{v_3(\OPTBundle_3)-v_3(\setaside_3)-\threshold}{5}$; this is shown in more detail in the previous lemma. Assume without loss of generality that $v_3(\bundle_3'') > v_3(\bundle_3')$. Since $v_3(\bundle_3^*) \geq \frac{v_3(\bundle_3)}{2} - v_3(\setaside_3)$ in line \ref{line:x3*def} (due to round-robin producing an EF1 allocation as shown in \cite{CKMPSW2019}), and $v_3(\bundle_1'') < \threshold$ by case assumption ($m_3(B_1) < \threshold$ since we did not return in line \ref{line:return1-else}), we have that the final value of agent 3 for $\bundle_3^*$ after line \ref{line:agent1-remove} is $v_3(\bundle_3^*) > \frac{v_3(\bundle_3)}{2} - v_3(\setaside_3) - \threshold = \frac{v_3(\OPTBundle_3)-v_3(\setaside_3)-\threshold}{10} - v_3(\setaside_3) - \threshold = \frac{v_3(\OPTBundle_3)-11v_3(\setaside_3)-11\threshold}{10}$.

    Since at the end of Procedure \ref{alg:3agents} each agent can choose between their bundle returned by Procedure \ref{alg:3agents-else} and their set-aside item, we obtain the following $\NSW$ guarantee:
    \begin{align*}
        \NSW(\allocALG)^3 \geq& \max\left(\frac{v_1(\OPTBundle_1) - v_1(\setaside_1)}{2}, v_1(\setaside_1)\right) \cdot\\
        &\max\left(\frac{v_2(\OPTBundle_2)-v_2(\setaside_2)}{12}, v_2(\setaside_2)\right) \cdot\\
        &\max\left(\frac{v_3(\OPTBundle_3)-11v_3(\setaside_3)-11\threshold}{10}, v_3(\setaside_3)\right).
    \end{align*}

    Via an exhaustive case analysis deferred to the appendix, we show that $\NSW(\allocALG) \geq \sqrt[3]{\frac{1-11\threshold}{6 \cdot 13 \cdot 21}}$.
    
    We now show that the allocation is EFx. Agent 1 does not EFx envy agent 2 since they run $\twoalgoDB$ in line \ref{line:2agentalgo-12} and $v_1(\bundle_1^*) \geq v_1(\bundle_1')$. She also does not envy agent 3 since she can pick the most valuable budget-feasible bundle out of $\bundle_3^*$ in line~\ref{line:agent1-pickbest}.
    
    Without loss, assume that agent 2 removes $\bundle_3''$ from $\bundle_3$, i.e. $\bundle_3^* = \bundle_3'$ in line \ref{line:x3*def}. Note that $v_2(\bundle_2) \geq \vmax_2(\bundle_3 \setminus \set{g})$ for all $g \in \bundle_3$ by the outcome of $\twoalgoDB$ in line \ref{line:2agentalgo-23}. By Theorem \ref{thm:2agents}, running $\twoalgoDB$ in line \ref{line:2agentalgo-12} guarantees that $v_2(\bundle_2') \geq \frac{v_2(\bundle_2)}{2} \geq \frac{\vmax_2(\bundle_3 \setminus \set{g})}{2} \geq \vmax_2(\bundle_3' \setminus \set{g'}) = \vmax_2(\bundle_3^* \setminus \set{g'})$ for all $g' \in \bundle_3^*$, where the first inequality is true by case assumption since $v_2(\bundle_2) < v_2(\bundle_1)$ (the procedure did not return in line \ref{line:return2-else}). Thus, agent 2 does not EFx envy agent 3. If agent 1 picks $\bundle_1^* = \bundle_1''$, agent 2 does not EFx envy her since this bundle is a subset of $\bundle_3^*$, and otherwise, the outcome of $\twoalgoDB$ in line \ref{line:2agentalgo-12} handles the envy.
    
    Lastly, consider the envy of agent 3. By definition of Procedure~\ref{alg:3agents}, we have $m_3(B_1) < \threshold$, which implies that $v_3(\bundle_1^*) < \threshold$ and $v_3(\bundle_2') < \threshold$, since both $\bundle_1^*$ and $\bundle_2'$ cost at most $B_1$. Thus, for $\max(\frac{v_3(\OPTBundle_3)-11v_3(\setaside_3)-11\threshold}{10}, v_3(\setaside_3)) \geq \threshold$ for $\threshold \leq \frac{1}{35}$, the returned allocation is EFx for agent 3.
\end{proof}

We now present the main result of our paper - achieving an EFx allocation with constant $\NSW$ for three-agent instances.

\begin{theorem} \label{thm:EFx-3agents-constantNSW}
    For instances with 3 agents, Procedure $\threealgoDB$ returns a budget-feasible EFx allocation with constant efficiency.
\end{theorem}
\begin{proof}
    By Lemmas \ref{lemma:equalbudget}, \ref{lemma:return1-else}, \ref{lemma:return2-else}, \ref{lemma:return3-else} all the terminal points in the procedure return a budget-feasible EFx allocation. If any of the agents choose their set-aside item at the end of Procedure \ref{alg:3agents}, the final allocation is budget-feasible by the budget-feasibility constraint on each agent's set-aside item. The final allocation is also EFx since agents who choose their set-aside item weakly improve their value and so they do not EFx-envy any other agent and are not EFx-envied by any other agents since they have a single item.

    The efficiency guarantee is obtained by finding the value of $\threshold$ that maximizes the $\NSW$ guarantee implied by the above lemmas, subject to the envy and nonnegativity constraints. For $\threshold = \frac{1}{35}$, all the constraints are satisfied, and the $\NSW$ of the returned allocation at each termination point is at least $1/171$.
\end{proof}

%% file: conclusion.tex
The main open problem is whether a constant approximation of the maximum Nash welfare is achievable through EFx budget-feasible allocations for an arbitrary number of agents. In fact, this question is open even for EF1 budget-feasible allocations. Note that in very recent and independent work, Barman et al. \cite{BKSS2023} extended the result of Chaudhury et al. \cite{CKMS2021} to allow for generalized assignment constraints. This captures our setting and implies the existence of budget-feasible EFx allocations in which no agent envies the charity for an arbitrary number of agents. While their algorithm does not yield a bounded approximation of the maximum $\NSW$, we believe that a modification using some of the ideas introduced in our paper can be shown to achieve a constant approximation of the $\NSW$ objective for a constant number of agents and, in general, a $O(n)$ approximation. Specifically, in the modified algorithm each agent initially sets aside the most valuable item from the bundle they receive in the maximum $\NSW$ allocation (if such bundle is not empty) and then a budget-feasible EFx allocation (with respect to the charity as well) is computed on the remaining goods using the \texttt{ComputeFEFx} procedure described in \cite{BKSS2023}. Lastly, each agent can pick their preferred set of goods between the bundle allocated to them by \texttt{ComputeFEFx} and the single good that they set aside in the first step. Similarly to our analysis in this paper, setting aside a good for each agent acts as a safety net against situations in which a very valuable good for one agent is allocated to another agent or to charity, which would still achieve EFx but can severely reduce efficiency. However, even if this approach could guarantee a linear approximation to the maximum $\NSW$ objective, it does not look like it could be used to achieve a constant approximation (or better than $O(n)$), which remains an interesting open problem.

Another interesting direction is to study the approximation achievable via budget-feasible envy-free allocations of \emph{divisible} goods. Although maximizing the $\NSW$ with divisible goods is always envy-free for (unrestricted) additive valuations, this is not true in the presence of budget constraints. For example, consider the following instance with two agents, both with a budget of 1, and two goods with cost $c(1) = c(2) = 1$, where the valuations are $v_1(1) = 1, v_1(2) = 0$, and $v_2(1) = 0.6, v_2(2) = 0.4$. The maximum $\NSW$ allocation gives agent 1 the whole good 1 and agent 2 the whole good 2 for a $\NSW$ of $\sqrt{0.4}$. However, the best envy-free allocation in terms of $\NSW$ gives agent 1 a 3/4 fraction of good 1 and agent 2 a 1/4 fraction of good 1 and a 3/4 fraction of good 2. This allocation has a $\NSW$ of $\sqrt{0.3375}$. Therefore, this instance yields no better than a 0.918 approximation of the maximum Nash welfare. Similarly to the indivisible case, the work of Barman et al. \cite{BKSS2023} provides an algorithm that returns a budget-feasible envy-free allocation (with respect to charity as well) for instances with divisible items. This implies an $O(n)$ approximation to the maximum $\NSW$ but it would be interesting to seek either a constant factor approximation procedure or provide an impossibility result for such a guarantee.

%% file: appendix.tex
\lemmaEqualBudget*
\begin{proof}
    Note that Procedure \ref{alg:3agents-equalbudget} returns a budget-feasible EFx allocation by Lemma~\ref{lemma:3agents-equal}. We now look at the agents' value guarantees. By Lemma~\ref{lemma:3agents-equal}, one agent $i$ gets a value of at least $\frac{v_i(\OPTBundle_i)}{9} - \frac{v_i(\setaside_i)}{3}$, another agent $j$ gets at least $\frac{v_j(\OPTBundle_j)}{9} - \frac{2v_j(\setaside_j)}{3}$, while the third agent $k$ gets at least $\frac{v_k(\OPTBundle_k)}{9} - v_k(\setaside_k)$. Since each agent $i$ can choose her set-aside item $\setaside_i$ at the end, it must be that $v_i(\bundle_i) \geq \max(v_i(\allocALGBundle_i), v_i(\setaside_i))$. By Lemma \ref{lemma:preprocess-remainingval}, $v_i(\OPTBundle_i)$ can take three different values based on which bundle the set-aside items were initially part of and later removed from. Note that since line \ref{line:reduce-budgets} in $\threealgoDB$ reduces the budgets of agents 2 and 3 to $B_1$, we can only guarantee that $v_2(\OPTBundle_2) \geq \threshold$ and $v_3(\OPTBundle_3) \geq \threshold$ through their monopoly values. Below, we compute all the possible distinct allocations (irrespective of agent ids, since they lead to the same $\NSW$ guarantee by symmetry).

    \begingroup
    \allowdisplaybreaks
    \begin{align*}
        \NSW(\allocALG)^3 \geq& \max\left(\frac{1-3v_1(\setaside_1)}{9} - v_1(\setaside_1), v_1(\setaside_1)\right) \cdot \max\left(\frac{\threshold}{9} - \frac{2v_2(\setaside_2)}{3}, v_2(\setaside_2)\right) \cdot\\
        &\max\left(\frac{\threshold}{9} - \frac{v_3(\setaside_3)}{3}, v_3(\setaside_3)\right)\\
                       \geq& \frac{1}{21} \cdot \frac{\threshold}{15} \cdot \frac{\threshold}{12}\\
        \NSW(\allocALG)^3 \geq& \max\left(\frac{1-3v_1(\setaside_1)}{9} - \frac{2v_1(\setaside_1)}{3}, v_1(\setaside_1)\right) \cdot \max\left(\frac{\threshold}{9} - v_2(\setaside_2), v_2(\setaside_2)\right) \cdot\\
        &\max\left(\frac{\threshold}{9} - \frac{v_3(\setaside_3)}{3}, v_3(\setaside_3)\right)\\
                       \geq& \frac{1}{18} \cdot \frac{\threshold}{18} \cdot \frac{\threshold}{12}\\
        \NSW(\allocALG)^3 \geq& \max\left(\frac{1-3v_1(\setaside_1)}{9} - \frac{v_1(\setaside_1)}{3}, v_1(\setaside_1)\right) \cdot
        \max\left(\frac{\threshold}{9} - \frac{2v_2(\setaside_2)}{3}, v_2(\setaside_2)\right) \cdot\\
        &\max\left(\frac{\threshold}{9} - v_3(\setaside_3), v_3(\setaside_3)\right)\\
                       \geq& \frac{1}{15} \cdot \frac{\threshold}{15} \cdot \frac{\threshold}{18}\\
        \NSW(\allocALG)^3 \geq& \max\left(\frac{1-2v_1(\setaside_1)}{9} - v_1(\setaside_1), v_1(\setaside_1)\right) \cdot
        \max\left(\frac{\threshold-v_2(\setaside_2)}{9} - \frac{2v_2(\setaside_2)}{3}, v_2(\setaside_2)\right) \cdot\\
        &\max\left(\frac{\threshold}{9} - \frac{v_3(\setaside_3)}{3}, v_3(\setaside_3)\right)\\
                       \geq& \frac{1}{20} \cdot \frac{\threshold}{16} \cdot \frac{\threshold}{12}\\
        \NSW(\allocALG)^3 \geq& \max\left(\frac{1-2v_1(\setaside_1)}{9} - v_1(\setaside_1), v_1(\setaside_1)\right) \cdot
        \max\left(\frac{\threshold-v_2(\setaside_2)}{9} - \frac{v_2(\setaside_2)}{3}, v_2(\setaside_2)\right) \cdot\\
        &\max\left(\frac{\threshold}{9} - \frac{2v_3(\setaside_3)}{3}, v_3(\setaside_3)\right)\\
                       \geq& \frac{1}{20} \cdot \frac{\threshold}{13} \cdot \frac{\threshold}{15}\\
        \NSW(\allocALG)^3 \geq& \max\left(\frac{1-2v_1(\setaside_1)}{9} - \frac{2v_1(\setaside_1)}{3}, v_1(\setaside_1)\right) \cdot
        \max\left(\frac{\threshold-v_2(\setaside_2)}{9} - \frac{v_2(\setaside_2)}{3}, v_2(\setaside_2)\right) \cdot\\
        &\max\left(\frac{\threshold}{9} - v_3(\setaside_3), v_3(\setaside_3)\right)\\
                       \geq& \frac{1}{17} \cdot \frac{\threshold}{13} \cdot \frac{\threshold}{18}\\
        \NSW(\allocALG)^3 \geq& \max\left(\frac{1-2v_1(\setaside_1)}{9} - \frac{2v_1(\setaside_1)}{3}, v_1(\setaside_1)\right) \cdot
        \max\left(\frac{\threshold-v_2(\setaside_2)}{9} - v_2(\setaside_2), v_2(\setaside_2)\right) \cdot\\
        &\max\left(\frac{\threshold}{9} - \frac{v_3(\setaside_3)}{3}, v_3(\setaside_3)\right)\\
                       \geq& \frac{1}{17} \cdot \frac{\threshold}{19} \cdot \frac{\threshold}{12}\\
        \NSW(\allocALG)^3 \geq& \max\left(\frac{1-2v_1(\setaside_1)}{9} - \frac{v_1(\setaside_1)}{3}, v_1(\setaside_1)\right) \cdot
        \max\left(\frac{\threshold-v_2(\setaside_2)}{9} - \frac{2v_2(\setaside_2)}{3}, v_2(\setaside_2)\right) \cdot\\
        &\max\left(\frac{\threshold}{9} - v_3(\setaside_3), v_3(\setaside_3)\right)\\
                       \geq& \frac{1}{14} \cdot \frac{\threshold}{16} \cdot \frac{\threshold}{18}\\
        \NSW(\allocALG)^3 \geq& \max\left(\frac{1-2v_1(\setaside_1)}{9} - \frac{v_1(\setaside_1)}{3}, v_1(\setaside_1)\right) \cdot
        \max\left(\frac{\threshold}{9} - \frac{2v_2(\setaside_2)}{3}, v_2(\setaside_2)\right) \cdot\\
        &\max\left(\frac{\threshold-v_3(\setaside_3)}{9} - v_3(\setaside_3), v_3(\setaside_3)\right)\\
                       \geq& \frac{1}{14} \cdot \frac{\threshold}{15} \cdot \frac{\threshold}{19}\\
        \NSW(\allocALG)^3 \geq& \max\left(\frac{1-v_1(\setaside_1)}{9} - v_1(\setaside_1), v_1(\setaside_1)\right) \cdot
        \max\left(\frac{\threshold-v_2(\setaside_2)}{9} - \frac{2v_2(\setaside_2)}{3}, v_2(\setaside_2)\right) \cdot\\
        &\max\left(\frac{\threshold-v_3(\setaside_3)}{9} - \frac{v_3(\setaside_3)}{3}, v_3(\setaside_3)\right)\\
                       \geq& \frac{1}{19} \cdot \frac{\threshold}{16} \cdot \frac{\threshold}{13}
    \end{align*}
    \endgroup
    
    For all of the inequalities, it holds that $\NSW(\alloc) \geq \sqrt[3]{\frac{\threshold^2}{15 \cdot 15 \cdot 18}}$.
\end{proof}

\lemmaCaseOne*
\begin{proof}
        First, notice that $\bundle_1 = \vargmax_1(M)$ is budget-feasible for agent 1 by definition. Since $\bundle_2$ and $\bundle_3$ result from running the $\twoalgoDB$ procedure, they are budget-feasible for agents 2 and 3 by Theorem \ref{thm:2agents} and so the returned allocation is budget-feasible.
    
    We now look at each agent's value guarantee. Note that agent 1 gets her monopoly bundle, so $v_1(\bundle_1) \geq v_1(\OPTBundle_1)$. In line~\ref{line:2agentalgo-23} agents 2 and 3 run $\twoalgoDB((\bar{\bundle}_2,\bar{\bundle}_3), M \setminus (\bundle_1 \cup \bar{\bundle}_2 \cup \bar{\bundle}_3))$, resulting in an EFx allocation $(\bundle_2, \bundle_3)$ for agents 2 and 3. We note that neither agent 2 nor 3 envies $M\setminus(\bundle_1 \cup \bar{\bundle}_2 \cup \bar{\bundle}_3)$ or otherwise they can get the envied subset resulting in an allocation with higher $\NSW$, contradicting the optimality of the allocation in line~\ref{line:maxNSW23}. Let $R$ be the set of unallocated items and $R'$ be the set of left-out items (as defined in Lemma~\ref{lemma:left-out}) after running $\twoalgoDB$ in line~\ref{line:2agentalgo-23}. We then get the following five inequalities:
    \begin{enumerate}
        \item $v_2(\bundle_2) \geq \vmax_2(\bundle_3 \setminus \set{g}) \quad \forall g \in \bundle_3$ \label{eq:agent2val-1-appendix}
        \item $v_2(\bundle_2) \geq \vmax_2(R)$ \label{eq:agent2val-2-appendix}
        \item $v_2(\bundle_2) \geq \vmax_2(R') \quad$ if $R' \subset \bundle_2$ \label{eq:agent2val-leftout-appendix}
        \item $v_2(\bundle_2) > \vmax_2(M\setminus(\bundle_1 \cup \bar{\bundle}_2 \cup \bar{\bundle}_3))$
        \item $v_2(\OPTBundle_2) \leq v_2(\bundle_1) + v_2(\bundle_2) + \vmax_2(\bundle_3) + \vmax_2(R) + \vmax_2(R') + \vmax_2(M\setminus(\bundle_1 \cup \bar{\bundle}_2 \cup \bar{\bundle}_3))$ \label{eq:agent2val-3-appendix}.
    \end{enumerate}
    Inequalities~(\ref{eq:agent2val-1-appendix}) and (\ref{eq:agent2val-2-appendix}) follow from the properties of $\twoalgoDB$, and inequality~(\ref{eq:agent2val-leftout-appendix}) follows from Lemma~\ref{lemma:left-out}. Inequality (\ref{eq:agent2val-3-appendix}) is true since the union of the bundles in the right-hand side is $M$. Since $v_2(\bundle_1) < \threshold$ by case assumption ($m_2(B_1) < \threshold$ by the if statement in line \ref{line:if-1}), $\vmax_2(g) \leq v_2(\setaside_2)$ by Lemma \ref{lemma:preprocess-maxval} and $\vmax_2(\bundle_3 \setminus \set{g}) \geq \vmax_2(\bundle_3) - \vmax_2(g)$ for all $g \in \bundle_3$ by Lemma \ref{lemma:subadditive}, we get that $v_2(\bundle_2) > \frac{v_2(\OPTBundle_2)-v_2(\setaside_2)-\threshold}{5}$.

    Similarly, for agent 3, we get the following five inequalities:
    \begin{itemize}
        \item $v_3(\bundle_3) \geq \vmax_3(\bundle_2 \setminus \set{g}) \quad \forall g \in \bundle_2$
        \item $v_3(\bundle_3) \geq \vmax_3(R)$
        \item $v_3(\bundle_3) \geq \vmax_3(R') \quad$ if $R' \subset \bundle_3$
        \item $v_3(\bundle_3) > \vmax_3(M\setminus(\bundle_1 \cup \bar{\bundle}_2 \cup \bar{\bundle}_3))$
        \item $v_3(\OPTBundle_3) \leq v_3(\bundle_1) + \vmax_3(\bundle_2) + v_3(\bundle_3) + \vmax_3(R) + \vmax_3(M\setminus(\bundle_1 \cup \bar{\bundle}_2 \cup \bar{\bundle}_3))$,
    \end{itemize}
    Thus, $v_3(\bundle_3) > \frac{v_3(\OPTBundle_3)-v_3(\setaside_3)-\threshold}{5}$, similarly as above and by case assumption which implies $v_3(\bundle_1) < \threshold$. 

    We obtain that the $\NSW$ guarantee is the following:
    \begin{align*}
        \NSW(\allocALG)^3 \geq& \max\left(v_1(\OPTBundle_1), v_1(\setaside_1)\right) \cdot\\
        &\max\left(\frac{v_2(\OPTBundle_2)-v_2(\setaside_2)-\threshold}{5}, v_2(\setaside_2)\right) \cdot\\
        &\max\left(\frac{v_3(\OPTBundle_3)-v_3(\setaside_3)-\threshold}{5}, v_3(\setaside_3)\right).
    \end{align*}
    
    Note that at least one of agents 2 and 3 will strictly not envy the other after running $\twoalgoDB$, however, to simplify the analysis, we just use the fact that they do not EFx-envy each other. Also, $R'$ can only be a subset of one of $\bundle_2$ and $\bundle_3$, whereas for simplicity we include both cases simultaneously.

    We now compute all the possible distinct allocations:
    \begingroup
    \allowdisplaybreaks
    \begin{align*}
        \NSW(\allocALG)^3 \geq& \max\left(1-3v_1(\setaside_1), v_1(\setaside_1)\right) \cdot
        \max\left(\frac{1-v_2(\setaside_2)-\threshold}{5}, v_2(\setaside_2)\right) \cdot
        \max\left(\frac{1-v_3(\setaside_3)-\threshold}{5}, v_3(\setaside_3)\right)\\
                       \geq& \frac{1}{4} \cdot \frac{1-\threshold}{6} \cdot \frac{1-\threshold}{6}\\
        \NSW(\allocALG)^3 \geq& \max\left(1-2v_1(\setaside_1), v_1(\setaside_1)\right) \cdot
        \max\left(\frac{1-2v_2(\setaside_2)-\threshold}{5}, v_2(\setaside_2)\right) \cdot
        \max\left(\frac{1-v_3(\setaside_3)-\threshold}{5}, v_3(\setaside_3)\right)\\
                       \geq& \frac{1}{3} \cdot \frac{1-\threshold}{7} \cdot \frac{1-\threshold}{6}\\
        \NSW(\allocALG)^3 \geq& \max\left(1-2v_1(\setaside_1), v_1(\setaside_1)\right) \cdot
        \max\left(\frac{1-v_2(\setaside_2)-\threshold}{5}, v_2(\setaside_2)\right) \cdot
        \max\left(\frac{1-2v_3(\setaside_3)-\threshold}{5}, v_3(\setaside_3)\right)\\
                       \geq& \frac{1}{3} \cdot \frac{1-\threshold}{6} \cdot \frac{1-\threshold}{7}\\
        \NSW(\allocALG)^3 \geq& \max\left(1-v_1(\setaside_1), v_1(\setaside_1)\right) \cdot
        \max\left(\frac{1-3v_2(\setaside_2)-\threshold}{5}, v_2(\setaside_2)\right) \cdot
        \max\left(\frac{1-v_3(\setaside_3)-\threshold}{5}, v_3(\setaside_3)\right)\\
                       \geq& \frac{1}{2} \cdot \frac{1-\threshold}{8} \cdot \frac{1-\threshold}{6}\\
        \NSW(\allocALG)^3 \geq& \max\left(1-v_1(\setaside_1), v_1(\setaside_1)\right) \cdot
        \max\left(\frac{1-v_2(\setaside_2)-\threshold}{5}, v_2(\setaside_2)\right) \cdot
        \max\left(\frac{1-3v_3(\setaside_3)-\threshold}{5}, v_3(\setaside_3)\right)\\
                       \geq& \frac{1}{2} \cdot \frac{1-\threshold}{6} \cdot \frac{1-\threshold}{8}\\
        \NSW(\allocALG)^3 \geq& \max\left(1, v_1(\setaside_1)\right) \cdot
        \max\left(\frac{1-4v_2(\setaside_2)-\threshold}{5}, v_2(\setaside_2)\right) \cdot
        \max\left(\frac{1-v_3(\setaside_3)-\threshold}{5}, v_3(\setaside_3)\right)\\
                       \geq& 1 \cdot \frac{1-\threshold}{9} \cdot \frac{1-\threshold}{6}\\
        \NSW(\allocALG)^3 \geq& \max\left(1, v_1(\setaside_1)\right) \cdot
        \max\left(\frac{1-v_2(\setaside_2)-\threshold}{5}, v_2(\setaside_2)\right) \cdot
        \max\left(\frac{1-4v_3(\setaside_3)-\threshold}{5}, v_3(\setaside_3)\right)\\
                       \geq& 1 \cdot \frac{1-\threshold}{6} \cdot \frac{1-\threshold}{9}
    \end{align*}
    \endgroup
    
    For all of the inequalities, it holds that $\NSW(\alloc) \geq \sqrt[3]{\frac{(1-\threshold)^2}{4 \cdot 6 \cdot 6}}$.

    For the envy, note that agent 1 does not envy agents 2 or 3 because she picked her monopoly bundle with respect to her budget. Agents 2 and 3 do not EFx envy each other by the property of $\twoalgoDB$. For $\threshold \leq \frac{1-\threshold}{9}$, they do not envy agent 1 since their value for $\bundle_1$ is at most $\threshold$ and they each get at least $\frac{1-\threshold}{9}$ (in the case where all the set-aside items were initially part of their maximum $\NSW$ bundle).
\end{proof}

\lemmaCaseTwo*
\begin{proof}
   Note that after line \ref{line:swap-agents}, the algorithm guarantees that $m_2(B_1) \geq \threshold$ and $m_3(B_1) < \threshold$. Since we do not perform any operations on the bundles as compared to their state in line \ref{line:return1-else}, the budget-feasibility and EFx guarantees for all agents remain the same as in Lemma \ref{lemma:return1-else}. Moreover, the individual value guarantees for agents 1 and 3 are the same as before. Only the value guarantees for agent 2 change, since now $v_2(\bundle_2) \geq v_2(\bundle_1)$. For agent 2's value, we get the following six inequalities:
    \begin{itemize}
        \item $v_2(\bundle_2) \geq v_2(\bundle_1)$
        \item $v_2(\bundle_2) \geq \vmax_2(\bundle_3 \setminus \set{g}) \quad \forall g \in \bundle_3$
        \item $v_2(\bundle_2) \geq \vmax_2(R)$
        \item $v_2(\bundle_2) \geq \vmax_2(R') \quad$ if $R' \subset \bundle_2$
        \item $v_2(\bundle_2) > \vmax_2(M\setminus(\bundle_1 \cup \bar{\bundle}_2 \cup \bar{\bundle}_3))$
        \item $v_2(\OPTBundle_2) \leq v_2(\bundle_1) + v_2(\bundle_2) + \vmax_2(\bundle_3) + \vmax_2(R) + \vmax_2(R') +  \vmax_2(M\setminus(\bundle_1 \cup \bar{\bundle}_2 \cup \bar{\bundle}_3))$
    \end{itemize}
    Similarly as in the previous proof, we get that $v_2(\bundle_2) \geq \frac{v_2(\OPTBundle_2)-v_2(\setaside_2)}{6}$.
    
    We obtain that the $\NSW$ guarantee is the following:
    \begin{align*}
        \NSW(\allocALG)^3 \geq& \max\left(v_1(\OPTBundle_1), v_1(\setaside_1)\right) \cdot\\
        &\max\left(\frac{v_2(\OPTBundle_2)-v_2(\setaside_2)}{6}, v_2(\setaside_2)\right) \cdot\\
        &\max\left(\frac{v_3(\OPTBundle_3)-v_3(\setaside_3)-\threshold}{5}, v_3(\setaside_3)\right).
    \end{align*}
    Then, similarly as in the full proof of Lemma \ref{lemma:return1-else}, we get the following possible allocations:

    \begingroup
    \allowdisplaybreaks
    \begin{align*}
        \NSW(\allocALG)^3 \geq& \max\left(1-3v_1(\setaside_1), v_1(\setaside_1)\right) \cdot
        \max\left(\frac{1-v_2(\setaside_2)}{6}, v_2(\setaside_2)\right) \cdot \max\left(\frac{1-v_3(\setaside_3)-\threshold}{5}, v_3(\setaside_3)\right)\\
                       \geq& \frac{1}{4} \cdot \frac{1}{7} \cdot \frac{1-\threshold}{6}\\
        \NSW(\allocALG)^3 \geq& \max\left(1-2v_1(\setaside_1), v_1(\setaside_1)\right) \cdot
        \max\left(\frac{1-2v_2(\setaside_2)}{6}, v_2(\setaside_2)\right) \cdot \max\left(\frac{1-v_3(\setaside_3)-\threshold}{5}, v_3(\setaside_3)\right)\\
                       \geq& \frac{1}{3} \cdot \frac{1}{8} \cdot \frac{1-\threshold}{6}\\
        \NSW(\allocALG)^3 \geq& \max\left(1-2v_1(\setaside_1), v_1(\setaside_1)\right) \cdot
        \max\left(\frac{1-v_2(\setaside_2)}{6}, v_2(\setaside_2)\right) \cdot \max\left(\frac{1-2v_3(\setaside_3)-\threshold}{5}, v_3(\setaside_3)\right)\\
                       \geq& \frac{1}{3} \cdot \frac{1}{7} \cdot \frac{1-\threshold}{7}\\
        \NSW(\allocALG)^3 \geq& \max\left(1-v_1(\setaside_1), v_1(\setaside_1)\right) \cdot
        \max\left(\frac{1-3v_2(\setaside_2)}{6}, v_2(\setaside_2)\right) \cdot \max\left(\frac{1-v_3(\setaside_3)-\threshold}{5}, v_3(\setaside_3)\right)\\
                       \geq& \frac{1}{2} \cdot \frac{1}{9} \cdot \frac{1-\threshold}{6}\\
        \NSW(\allocALG)^3 \geq& \max\left(1-v_1(\setaside_1), v_1(\setaside_1)\right) \cdot
        \max\left(\frac{1-v_2(\setaside_2)}{6}, v_2(\setaside_2)\right) \cdot \max\left(\frac{1-3v_3(\setaside_3)-\threshold}{5}, v_3(\setaside_3)\right)\\
                       \geq& \frac{1}{2} \cdot \frac{1}{7} \cdot \frac{1-\threshold}{8}\\
        \NSW(\allocALG)^3 \geq& \max\left(1-v_1(\setaside_1), v_1(\setaside_1)\right) \cdot
        \max\left(\frac{1-2v_2(\setaside_2)}{6}, v_2(\setaside_2)\right) \cdot \max\left(\frac{1-2v_3(\setaside_3)-\threshold}{5}, v_3(\setaside_3)\right)\\
                       \geq& \frac{1}{2} \cdot \frac{1}{8} \cdot \frac{1-\threshold}{7}\\
        \NSW(\allocALG)^3 \geq& \max\left(1, v_1(\setaside_1)\right) \cdot
        \max\left(\frac{1-4v_2(\setaside_2)}{6}, v_2(\setaside_2)\right) \cdot \max\left(\frac{1-v_3(\setaside_3)-\threshold}{5}, v_3(\setaside_3)\right)\\
                       \geq& 1 \cdot \frac{1}{10} \cdot \frac{1-\threshold}{6}\\
        \NSW(\allocALG)^3 \geq& \max\left(1, v_1(\setaside_1)\right) \cdot
        \max\left(\frac{1-v_2(\setaside_2)}{6}, v_2(\setaside_2)\right) \cdot \max\left(\frac{1-4v_3(\setaside_3)-\threshold}{5}, v_3(\setaside_3)\right)\\
                       \geq& 1 \cdot \frac{1}{7} \cdot \frac{1-\threshold}{9}
    \end{align*}
    \endgroup
    
    For all of the inequalities, it holds that $\NSW(\alloc) \geq \sqrt[3]{\frac{1-\threshold}{4 \cdot 7 \cdot 6}}$.
\end{proof}

\lemmaCaseThree*
\begin{proof}
    For budget-feasibility, note that $\bundle_1^*$ is budget-feasible for agent 1 since both $\bundle_1'$ and $\bundle_1''$ are budget-feasible for her. Agent 2 gets $\bundle_2' \subseteq \bundle_1$ by the outcome of $\twoalgoDB$ in line~\ref{line:2agentalgo-12}, and since $c(\bundle_1) \leq B_1 \leq B_2$, $\bundle_2'$ is budget-feasible for agent 2. Agent 3 also gets a budget-feasible bundle in $\bundle_3$ by the outcome of $\twoalgoDB$ in line~\ref{line:2agentalgo-23}, and since $\bundle_3^* \subseteq \bundle_3$, the returned allocation is budget-feasible.
    
    We now consider the individual value guarantees of the agents. Note that $\bundle_1$ is agent 1's monopoly bundle. By running $\twoalgoDB$ in line \ref{line:2agentalgo-12}, we get that $v_1(\bundle_1') \geq v_1(\bundle_2' \setminus \set{g})$ for all $g \in \bundle_2'$, and $v_1(R_1) = 0$, where $R_1$ is the set of unallocated items after running $\twoalgoDB$ here. Since $v_1(\bundle_1') + v_1(\bundle_2') = v_1(\bundle_1) \geq v_1(\OPTBundle_1)$ and agent 1 has the option to increase her value by picking $\bundle_1''$, we get $v_1(\bundle_1^*) \geq \frac{v_1(\OPTBundle_1) - v_1(\setaside_1)}{2}$. Note that there are no left-out items (as defined by Lemma~\ref{lemma:left-out}).

    For agent 2, we have $v_2(\bundle_2) \geq \vmax_2(\bundle_3 \setminus \set{g})$ for all $g \in \bundle_3$ from running $\twoalgoDB$ on agents 2 and 3 in line \ref{line:2agentalgo-23}. Since $v_2(\bundle_2) < v_2(\bundle_1)$ by case assumption and $v_2(\OPTBundle_2) \leq v_2(\bundle_1) + v_2(\bundle_2) + \vmax_2(\bundle_3) + \vmax_2(R_2) + \vmax_2(R_2') + \vmax_2(M\setminus(\bundle_1 \cup \bar{\bundle}_2 \cup \bar{\bundle}_3))$, where $R_2$ is the set of unallocated items and $R_2'$ is the set of left-out items after running $\twoalgoDB$ in line~\ref{line:2agentalgo-23}, we have $v_2(\bundle_1) > \frac{v_2(\OPTBundle_2)-v_2(\setaside_2)}{6}$; this is shown in more detail in the previous lemma. By the outcome of $\twoalgoDB$ in line \ref{line:2agentalgo-12}, we get $v_2(\bundle_2') \geq \frac{\vmax_2(\bundle_1)}{2} > \frac{v_2(\OPTBundle_2)-v_2(\setaside_2)}{12}$.
    
    Now consider the value of agent 3. Similarly, we know that $v_3(\OPTBundle_3) \leq v_3(\bundle_1) + \vmax_3(\bundle_2) + v_3(\bundle_3) + \vmax_3(R_2) + \vmax_3(R_2') + \vmax_3(M\setminus(\bundle_1 \cup \bar{\bundle}_2 \cup \bar{\bundle}_3))$, and $v_3(\bundle_1) < \threshold$ by case assumption, and $v_3(\bundle_3) \geq \vmax_3(\bundle_2 \setminus \set{g})$, for all $g \in \bundle_2$, and $v_3(\bundle_3) \geq \vmax_3(R_2)$ by the properties of $\twoalgoDB$. Thus, we get that $v_3(\bundle_3) > \frac{v_3(\OPTBundle_3)-v_3(\setaside_3)-\threshold}{5}$; this is shown in more detail in the previous lemma. Assume without loss of generality that $v_3(\bundle_3'') > v_3(\bundle_3')$. Since $v_3(\bundle_3^*) \geq \frac{v_3(\bundle_3)}{2} - v_3(\setaside_3)$ in line \ref{line:x3*def} (due to round-robin producing an EF1 allocation as shown in \cite{CKMPSW2019}), and $v_3(\bundle_1'') < \threshold$ by case assumption ($m_3(B_1) < \threshold$ since we did not return in line \ref{line:return1-else}), we have that the final value of agent 3 for $\bundle_3^*$ after line \ref{line:agent1-remove} is $v_3(\bundle_3^*) > \frac{v_3(\bundle_3)}{2} - v_3(\setaside_3) - \threshold = \frac{v_3(\OPTBundle_3)-v_3(\setaside_3)-\threshold}{10} - v_3(\setaside_3) - \threshold = \frac{v_3(\OPTBundle_3)-11v_3(\setaside_3)-11\threshold}{10}$.

    Since at the end of Procedure \ref{alg:3agents} each agent can choose between their bundle returned by Procedure \ref{alg:3agents-else} and their set-aside item, we obtain the following $\NSW$ guarantee:
    \begin{align*}
        \NSW(\allocALG)^3 \geq& \max\left(\frac{v_1(\OPTBundle_1) - v_1(\setaside_1)}{2}, v_1(\setaside_1)\right) \cdot\\
        &\max\left(\frac{v_2(\OPTBundle_2)-v_2(\setaside_2)}{12}, v_2(\setaside_2)\right) \cdot\\
        &\max\left(\frac{v_3(\OPTBundle_3)-11v_3(\setaside_3)-11\threshold}{10}, v_3(\setaside_3)\right).
    \end{align*}

    Similarly as in the full proof of Lemma \ref{lemma:return1-else}, we get the following possible allocations:

    \begingroup
    \allowdisplaybreaks
    \begin{align*}
        \NSW(\allocALG)^3 \geq& \max\left(\frac{1-4v_1(\setaside_1)}{2}, v_1(\setaside_1)\right) \cdot
        \max\left(\frac{1-v_2(\setaside_2)}{12}, v_2(\setaside_2)\right) \cdot\max\left(\frac{1-11v_3(\setaside_3)-11\threshold}{10}, v_3(\setaside_3)\right)\\
                      \geq& \frac{1}{6} \cdot \frac{1}{13} \cdot \frac{1-11\threshold}{21}\\
        \NSW(\allocALG)^3 \geq& \max\left(\frac{1-3v_1(\setaside_1)}{2}, v_1(\setaside_1)\right) \cdot
        \max\left(\frac{1-2v_2(\setaside_2)}{12}, v_2(\setaside_2)\right) \cdot\max\left(\frac{1-11v_3(\setaside_3)-11\threshold}{10}, v_3(\setaside_3)\right)\\
                      \geq& \frac{1}{5} \cdot \frac{1}{14} \cdot \frac{1-11\threshold}{21}\\
        \NSW(\allocALG)^3 \geq& \max\left(\frac{1-3v_1(\setaside_1)}{2}, v_1(\setaside_1)\right) \cdot
        \max\left(\frac{1-v_2(\setaside_2)}{12}, v_2(\setaside_2)\right) \cdot\max\left(\frac{1-12v_3(\setaside_3)-11\threshold}{10}, v_3(\setaside_3)\right)\\
                      \geq& \frac{1}{5} \cdot \frac{1}{13} \cdot \frac{1-11\threshold}{22}\\
        \NSW(\allocALG)^3 \geq& \max\left(\frac{1-2v_1(\setaside_1)}{2}, v_1(\setaside_1)\right) \cdot
        \max\left(\frac{1-3v_2(\setaside_2)}{12}, v_2(\setaside_2)\right) \cdot\max\left(\frac{1-11v_3(\setaside_3)-11\threshold}{10}, v_3(\setaside_3)\right)\\
                      \geq& \frac{1}{4} \cdot \frac{1}{15} \cdot \frac{1-11\threshold}{21}\\
        \NSW(\allocALG)^3 \geq& \max\left(\frac{1-2v_1(\setaside_1)}{2}, v_1(\setaside_1)\right) \cdot
        \max\left(\frac{1-v_2(\setaside_2)}{12}, v_2(\setaside_2)\right) \cdot\max\left(\frac{1-11v_3(\setaside_3)-11\threshold}{10}, v_3(\setaside_3)\right)\\
                      \geq& \frac{1}{4} \cdot \frac{1}{13} \cdot \frac{1-11\threshold}{23}\\
        \NSW(\allocALG)^3 \geq& \max\left(\frac{1-2v_1(\setaside_1)}{2}, v_1(\setaside_1)\right) \cdot
        \max\left(\frac{1-2v_2(\setaside_2)}{12}, v_2(\setaside_2)\right) \cdot\max\left(\frac{1-12v_3(\setaside_3)-11\threshold}{10}, v_3(\setaside_3)\right)\\
                      \geq& \frac{1}{4} \cdot \frac{1}{14} \cdot \frac{1-11\threshold}{22}\\
        \NSW(\allocALG)^3 \geq& \max\left(\frac{1-v_1(\setaside_1)}{2}, v_1(\setaside_1)\right) \cdot
        \max\left(\frac{1-4v_2(\setaside_2)}{12}, v_2(\setaside_2)\right) \cdot\max\left(\frac{1-11v_3(\setaside_3)-11\threshold}{10}, v_3(\setaside_3)\right)\\
                      \geq& \frac{1}{3} \cdot \frac{1}{16} \cdot \frac{1-11\threshold}{21}\\
        \NSW(\allocALG)^3 \geq& \max\left(\frac{1-v_1(\setaside_1)}{2}, v_1(\setaside_1)\right) \cdot
        \max\left(\frac{1-v_2(\setaside_2)}{12}, v_2(\setaside_2)\right) \cdot\max\left(\frac{1-14v_3(\setaside_3)-11\threshold}{10}, v_3(\setaside_3)\right)\\
                      \geq& \frac{1}{3} \cdot \frac{1}{13} \cdot \frac{1-11\threshold}{24}
    \end{align*}
    \endgroup
    
    For all of the inequalities, it holds that $\NSW(\alloc) \geq \sqrt[3]{\frac{1-11\threshold}{6 \cdot 13 \cdot 21}}$.
    
    We now show that the allocation is EFx. Agent 1 does not EFx envy agent 2 since they run $\twoalgoDB$ in line \ref{line:2agentalgo-12} and $v_1(\bundle_1^*) \geq v_1(\bundle_1')$. She also does not envy agent 3 since she can pick the most valuable budget-feasible bundle out of $\bundle_3^*$ in line~\ref{line:agent1-pickbest}.
    
    Without loss, assume that agent 2 removes $\bundle_3''$ from $\bundle_3$, i.e. $\bundle_3^* = \bundle_3'$ in line \ref{line:x3*def}. Note that $v_2(\bundle_2) \geq \vmax_2(\bundle_3 \setminus \set{g})$ for all $g \in \bundle_3$ by the outcome of $\twoalgoDB$ in line \ref{line:2agentalgo-23}. By Theorem \ref{thm:2agents}, running $\twoalgoDB$ in line \ref{line:2agentalgo-12} guarantees that $v_2(\bundle_2') \geq \frac{v_2(\bundle_2)}{2} \geq \frac{\vmax_2(\bundle_3 \setminus \set{g})}{2} \geq \vmax_2(\bundle_3' \setminus \set{g'}) = \vmax_2(\bundle_3^* \setminus \set{g'})$ for all $g' \in \bundle_3^*$, where the first inequality is true by case assumption since $v_2(\bundle_2) < v_2(\bundle_1)$ (the procedure did not return in line \ref{line:return2-else}). Thus, agent 2 does not EFx envy agent 3. If agent 1 picks $\bundle_1^* = \bundle_1''$, agent 2 does not EFx envy her since this bundle is a subset of $\bundle_3^*$, and otherwise, the outcome of $\twoalgoDB$ in line \ref{line:2agentalgo-12} handles the envy.
    
    Lastly, consider the envy of agent 3. By definition of Procedure~\ref{alg:3agents}, we have $m_3(B_1) < \threshold$, which implies that $v_3(\bundle_1^*) < \threshold$ and $v_3(\bundle_2') < \threshold$, since both $\bundle_1^*$ and $\bundle_2'$ cost at most $B_1$. Thus, for $\max(\frac{v_3(\OPTBundle_3)-11v_3(\setaside_3)-11\threshold}{10}, v_3(\setaside_3)) \geq \threshold$ for $\threshold \leq \frac{1}{35}$, the returned allocation is EFx for agent 3.
\end{proof}